\documentclass[review]{elsarticle}

\usepackage{lineno}
\usepackage{hyperref}
\biboptions{authoryear}

\usepackage{epsfig, amsmath, amsfonts, amssymb, amsthm,multirow,algorithm,algorithmic}
\usepackage{float}
\usepackage{diagbox}
\usepackage{graphicx}
\usepackage{setspace}

\usepackage{enumerate}
\usepackage{hyperref}
\hypersetup{citecolor=blue,colorlinks=true}

\theoremstyle{plain}
\newtheorem{thm}{Theorem}[section]

\newtheorem{proposition}[thm]{Proposition}

\newcommand{\bit}{\begin{itemize}}
\newcommand{\eit}{\end{itemize}}

\newcommand{\ben}{\begin{enumerate}}
\newcommand{\een}{\end{enumerate}}

\numberwithin{equation}{section}

\theoremstyle{plain}

\theoremstyle{remark}
\newtheorem*{remark}{Remark}

\newcommand{\x}{{\boldsymbol{x}}}
\newcommand{\X}{{\boldsymbol{X}}}

\newcommand{\btheta}{{\boldsymbol{\theta}}}

\DeclareMathOperator*{\argmin}{arg\,min}
\newcommand{\be}{\begin{equation}}
\newcommand{\ee}{\end{equation}}
\newcommand{\ba}{\begin{eqnarray}}
\newcommand{\ea}{\end{eqnarray}}
\newcommand{\bee}{\begin{equation*}}
\newcommand{\eee}{\end{equation*}}
\newcommand{\baa}{\begin{eqnarray*}}
	\newcommand{\eaa}{\end{eqnarray*}}

\begin{document}

\begin{frontmatter}

\title{Robust K-means Clustering using the Density Power Divergence Measure}

\author{Anirban Mondal}
\address{Case Western Reserve University, Cleveland, OH 44106, USA}

\author{Paromita Banerjee}
\address{John Carroll University, University Heights, OH, 44122, USA}

\author{Abhijit Mandal}
\address{The University of Texas at El Paso, El Paso, TX}

\begin{abstract} 

We introduce a robust clustering method, MK-means DPD, that estimates cluster centers and covariance matrices using density power divergence (DPD) measures combined with Mahalanobis distance, making it resistant to outliers and adaptable to heterogeneous, elliptical clusters, unlike the classical K-means algorithm. Since Mahalanobis distance-based K-means lacks a general convergence guarantee, we further introduce a convergent variant, Density-Consistent MK-means DPD (DC-MK-means DPD), which redefines the cluster assignment step in terms of a pointwise DPD loss. We prove a formal theorem establishing that the resulting algorithm converges in a finite number of steps. We also propose two new robust internal evaluation indexes, a Median Davies-Bouldin Index and a Trimmed Calinski-Harabasz Index, to ensure that performance comparisons are not themselves distorted by outliers. The efficacy of the proposed methods is demonstrated on simulated data, showing superiority over existing methods, and on two real datasets: Iris data, to identify similar species, and Covid-19 case fatality rate and infection rate data for countries worldwide, examining the resulting clusters' geographic and socio-economic patterns.

\end{abstract}

\begin{keyword} 
Clustering, K-means, Density power divergence, Mahalanobis distance 
\end{keyword}
		
\end{frontmatter}

\section{Introduction}

Cluster analysis, a statistical technique employed for discerning groups exhibiting similar patterns in a dataset, stands as a widely used unsupervised learning method in the machine learning community. It partitions data into groups that are internally homogeneous and mutually heterogeneous \citep{sharma1995applied}. Its applications span diverse fields, including education, finance, agriculture, genetics, and social network analysis, underscoring its versatility in uncovering patterns and relationships within data. As the volume of unlabeled data continues to grow, clustering has become, and will remain, one of the most significant techniques in machine learning.

A variety of clustering algorithms exist in the literature, spanning partition-based methods such as K-means \citep{MacQueen1967} and K-medoids \citep{kaufman2009finding}, which divide the dataset into a predetermined number of clusters; hierarchical methods, such as Ward's method \citep{Ward63} and Single Linkage clustering, which build a nested hierarchy of clusters through recursive merging or splitting; density-based methods, such as DBSCAN \citep{DBSCAN}, which group points based on local density; distribution-based and model-based methods, which model the data as arising from an underlying probability distribution; grid-based methods, which partition the data space into cells; and constraint-based methods, which incorporate user-specified constraints into the clustering process. Each approach has its own advantages and disadvantages, and the choice of algorithm depends on the characteristics of the data and the problem at hand \citep{tan2006introduction}.
 
This article explores various adaptations of the K-means clustering algorithm, the most widely recognized partition-optimization algorithm. The K-means algorithm segments observations into K clusters by minimizing within-cluster variances, with each observation being assigned to the cluster having the nearest cluster center.

An important consideration in applying the K-means algorithm is the suitability of a particular distance measure to accommodate clusters of specific configurations. While the commonly used Euclidean distance measure is effective for clusters with roughly spherical, homogeneous covariance matrices, it proves inadequate for clusters with elliptic shapes or heterogeneous variances across different dimensions—frequent occurrences in real-life applications. In such cases, the utilization of the Mahalanobis distance is advantageous, as this distance metric takes into account the variability within clusters. The incorporation of the variance-covariance matrix in the distance metric enhances flexibility for elliptical shapes and heterogeneity across multiple variables in the data. It is noteworthy that Mahalanobis is a generalization of Euclidean distance, coinciding with the latter when variables are homogeneous and uncorrelated. However, the use of the Mahalanobis distance necessitates the estimation of additional parameters of the covariance matrix \citep{Gnanadesikan1993}, which can pose a challenge, especially for high-dimensional data.

A shared limitation of the traditional Euclidean distance-based K-means algorithm and the Mahalanobis distance-based K-means algorithm is their susceptibility to the impact of outliers. The presence of outliers significantly influences both the determination of cluster centers and the computation of distance metrics. The distorted estimation of cluster centers and distance metrics, arising from the presence of outliers, often results in inappropriate clustering allocations for both the Euclidean distance-based K-means algorithm. Similarly, the Mahalanobis distance-based K-means algorithm (MK-means) exhibits poor performance in datasets contaminated by outliers due to distorted estimations of cluster centers and covariance matrices.

To address this challenge, we propose a robust K-means clustering method where the estimation of cluster centers and covariance matrices relies on density power divergence (DPD) measures. The Mahalanobis distance serves as the distance metric, accommodating elliptical clusters and heterogeneous correlated variables within the dataset. The robustness to extreme values inherent in the DPD measure-based estimation method enhances the resilience of the proposed method to outliers in the data. Adjusting the tuning parameter in the DPD measure according to the degree of outlier contamination allows for obtaining optimal robust estimates for cluster centers, covariance matrices, and the corresponding clusters.

A further contribution of this article concerns the convergence behavior of the proposed algorithm. Mahalanobis distance-based K-means clustering, on which our method is built, is known to lack a general convergence guarantee, since unconstrained re-estimation of the cluster covariance matrices can prevent the underlying distance objective from decreasing monotonically. To address this, we introduce a convergent variant of the proposed algorithm, which we call Density-Consistent MK-means DPD (DC-MK-means DPD), in which the cluster assignment step is redefined in terms of a pointwise density power divergence loss rather than the raw Mahalanobis distance. We prove that this modification restores a well-defined joint objective function that is non-increasing at every iteration, and establish a formal convergence theorem showing that DC-MK-means DPD converges in a finite number of steps, analogous to the classical convergence guarantees of Lloyd's and MacQueen's algorithms for Euclidean K-means.

As a further contribution, we also address a related but often overlooked issue: the internal metrics conventionally used to evaluate clustering performance, such as the silhouette index, the Davies-Bouldin Index (DBI), and the Calinski-Harabasz Index (CHI), are themselves sensitive to outliers, since they rely on averages and squared distances computed around cluster means. In addition to using existing robust alternatives, namely trimmed $R^2$ and the medoid-based silhouette index (PAMSIL), we propose two new robust internal evaluation indexes based on DBI and CHI, referred to as the Median Davies-Bouldin Index and the Trimmed Calinski-Harabasz Index. The former uses medians instead of averages to compute both the cluster centers and the within-cluster dispersion, while the latter uses trimming, preserving the variance-decomposition identity on which the Calinski-Harabasz Index depends. These robust indexes are used in place of their classical counterparts throughout the performance evaluations in this article, ensuring that the reported comparisons are not themselves distorted by the very outliers the proposed clustering method is designed to withstand.

The article is structured as follows: In Section 2, we introduce the K-means algorithm and the Mahalanobis distance-based K-means (MK-means) algorithm. Section 3 introduces the DPD measure and its corresponding parameter estimation method, along with our proposed robust K-means algorithm utilizing Mahalanobis distance and DPD measure, as well as the convergent DC-MK-means DPD variant and its associated convergence theorem. Section 4 introduces the internal and external evaluation metrics used to assess clustering performance, including the two new robust internal indexes proposed in this article, before applying the proposed methods to simulated data and comparing their performance with existing methods. In Section 5, we apply the methodology to two real-world applications. Section 6 concludes the article with a concise summary and discussion of the methods and their applications, along with suggestions for future research in this context.

\section{K-means using Mahalanobis distance}
In this section, we describe the widely used K-means algorithm where the distance metric used is the Mahalanobis distance. Let $\x_1, \x_2, \cdots, \x_n$ be $n$ observations in a p-dimensional space. Also, assume that there are $K$ clusters. The classical K-means algorithm minimize $\Delta(C_1, C_2, \cdots,C_K)= \sum_{k=1}^K\sum_{i\in C_k}||\x_i- \bar{\x_k}||^2$, where $C_k$ denotes the set of points which are assigned to the $k$-th cluster, $\bar{x_k}=\frac{1}{n_k}\sum_{i\in C_k}\x_i$, and $n_k$ is the number of points in the $k$-th cluster. The K-means algorithm using the Mahalanobis distance, from here-in referred to as MK-means, is given by the following steps.

\begin{algorithm}[H]
\caption{K-means using Mahalanobis distance (MK-means)}
\label{MK-means}
\begin{enumerate}
    \item Set initial values of $K$ cluster centers $(\hat{\mu}_{1}, \cdots, \hat{\mu}_{K})$ 
    \item Assign $\x_1$ to the cluster $k$ if the euclidean distance between $\x_1$ and $\hat{\mu}_{k}$ is the smallest. Repeat this for $\x_2, \x_3, \cdots \x_n$. Let $C_{0k}$ denote the set of points which are assigned to the $k$-th cluster, $k=1,2, \cdots, K$. 
    \item Estimate $\hat{\Sigma}_{k} = \frac{1}{n} \sum_{i\in C_{0k}}(\x_i-\hat{\mu}_{k})(\x_i-\hat{\mu}_{k})^T$, $k=1,2, \cdots, K$.
    \item Calculate the distance of $\x_i$ with the cluster center $\hat{\mu}_{k}$ by the Mahalanobis distance given by  $d(\x_i, \hat{\mu}_{k})=(\x_i-\hat{\mu}_{k})\hat{\Sigma}_{k}^{-1}(\x_i-\hat{\mu}_{k})^T$, $k=1,2, \cdots, K$. Assign $\x_i$ to cluster $m$ if $\min_{k}(d(\X_i, \hat{\mu}_{k}))=d(\X_i, \hat{\mu}_{m}))$. Repeat this for $i=1,2,\cdots,n$ and obtain the new clusters. Let $C_k$ denote the set of points assigned to the $k$-th cluster, $k=1,2,\cdots,K$.
    \item Update $\hat{\mu}_{k}$ and  $\hat{\Sigma}_{k}$ by the maximum likelihood estimators: $\hat{\mu}_{k}=\frac{1}{n_k}\sum_{i\in C_k}\x_i$, and $\hat{\Sigma}_{k}=\frac{1}{n_k}\sum_{i\in C_k}(\x_i-\hat{\mu}_{k})(\x_i-\hat{\mu}_{k})^T$.
    \item Repeat Steps 4 and 5 until the changes in $\hat\mu_k$ and $\hat\Sigma_k$ between successive iterations fall below a pre-specified tolerance, or until a pre-specified maximum number of iterations is reached.
\end{enumerate}
\end{algorithm}

Utilizing Mahalanobis distance within the K-means clustering framework presents a significant advantage over the traditional K-means with Euclidean distance, especially when dealing with clusters exhibiting heterogeneity and elliptical shapes. The outcomes of both algorithms align when clusters are homogeneous and spherical. However, challenges emerge in the selection of initial cluster centers and covariances, critical for the successful execution of the algorithm. A notable concern is the potential impact of incorrectly specified initial cluster covariance, leading to inappropriate clusters. \cite{Melnykov} proposed an algorithm to address this challenge, involving the identification of a core group of points with a high concentration of neighbors to serve as a representative 'core' for the selected cluster. This core group aids in providing an initial estimate of the covariance matrix, and we adopt a similar approach in our proposed algorithm.

Another significant drawback of the aforementioned algorithm is its susceptibility to outliers, influencing the maximum likelihood estimators (MLE) of cluster centers and covariance matrices in Step 3. This sensitivity to extreme observations hampers the algorithm's performance. To tackle this issue, we introduce a robust algorithm that employs density power divergence (DPD)-based estimators instead of MLE-based estimators. The details of this innovative algorithm are expounded upon in the subsequent section.

\section{Robust K-means algorithm}
Before we introduce our algorithm, we first introduce the DPD measure, some of its robustness properties, and the estimation of parameters from a dataset using DPD.
\subsection{The DPD Measure and Parameter Estimation}
The density power divergence (DPD) measure between the model density $f_\btheta$ with $\btheta \in \Theta$ and the empirical (or true) density $g$ is defined as
	\be 
	d_\alpha(f_\btheta, g) = 
	\left\{
	\begin{array}{ll}
		\int_y\left\{ f^{1+\alpha}_\btheta(y)-\left( 1+\frac{1}{\alpha}\right) f^{\alpha }_\btheta(y)g(y)+
		\frac{1}{\alpha}g^{1+\alpha}(y)\right\} dy, & \text{for}\mathrm{~}\alpha>0, \\%
		[2ex]
		\int_y g(y)\log\left( \displaystyle\frac{g(y)}{f_\btheta(y)}\right) dy, & \text{for}
		\mathrm{~}\alpha=0,
	\end{array}
	\right. 
	\label{dpd}
	\ee
	where $\alpha$ is a tuning parameter \citep{MR1665873}. 	For $\alpha=0$, the DPD is obtained as a limiting case of $\alpha \rightarrow 0^+$; and the measure is called the Kullback-Leibler divergence. 
	Given a parametric model, the estimation of $\btheta$ involves minimizing the DPD measure with respect to $\btheta$ over its parametric space $\Theta$. This estimator is referred to as the Minimum Power Divergence Estimator (MDPDE). For $\alpha=0$, it is analogous to maximizing the log-likelihood function, making the Maximum Likelihood Estimator (MLE) a specific instance of MDPDE. The tuning parameter $\alpha$ governs the balance between the efficiency and robustness of MDPDE—robustness increases with higher $\alpha$ but comes at the cost of decreased efficiency. The DPD and its associated techniques are extensively explored in statistical inference for parameter estimation and testing, as evidenced in various studies, such as \cite{MR3011625,combinedJSPI, MR3435166,  basu2017wald,composite2017,metrika,emptyCell,inlier,gof}. 

Here we apply MDPDE within the realm of cluster analysis, where the primary aim is to assemble similar data into clusters based on shared characteristics. MDPDE utilizes the concept of density power divergence, a metric gauging the dissimilarity between two probability distributions, to estimate parameters linked to the underlying distribution within each cluster. In this study, particular attention is given to the utilization of the normal distribution. Parameters of the normal distribution are estimated using DPD measure. This process effectively captures essential features of data distributions within clusters. By incorporating MDPDE with the normal distribution, we enhance our capability to achieve accurate and robust clustering results. 

Now we briefly describe the empirical methods on estimating the parameters ($\btheta$) assuming we have $m$ samples from the true distribution, viz. $y_1, y_2, \ldots, y_m$.  

Note that the third term of divergence of equation \eqref{dpd} is free of $\theta$, thus the power divergence estimator of $\theta$ can be found by minimizing:
\begin{align}
V_{\alpha, \btheta}({\bf y})=\int f_\btheta^{1+\alpha}(y) d y-\left(1+\frac{1}{\alpha}\right) \frac{1}{m} \sum_{i=1}^m f_\btheta^{(\alpha)}\left(y_i\right),
\label{normal_eqation}
\end{align}
where $f_\btheta(y)$ is the pdf of the Gaussian distribution, as is given by $f_\btheta(y)=\frac{1}{\sqrt{2 \pi} \sigma} \exp \left[-\frac{1}{2}\left(\frac{y-u}{\sigma}\right)^2\right]$. Using this pdf, the first term in Equation \eqref{normal_eqation} can be computed as $\int_y f_\btheta^{1+\alpha}(y) d y=(2 \pi)^{-\frac{\alpha}{2}} \sigma^\alpha(1+\alpha)^{-\frac{1}{2}}$.

So by \eqref{normal_eqation} the parameter estimates can be obtained by minimizing

\begin{align}
V_{\alpha, \mu, \sigma}({\bf y})  =(2 \pi)^{-\alpha / 2} \sigma^{-\alpha}(1+\alpha)^{-\frac{1}{2}}\left[1-\frac{(1+\alpha)^{3 / 2}}{m \alpha} \sum_{i=1}^m \exp \left[-\frac{\alpha}{2}\left(\frac{y_i-\mu}{\sigma}\right)^2\right]\right].
\label{normal_equation1}
\end{align}

For a given data $y_1, y_2, \ldots, y_m$, the above quantity is minimized with respect to the mean $\mu$ and variance $\sigma^2$ to obtain the MDPDE. The tuning parameter $\alpha$ regulates the balance between efficiency and robustness in the MDPDE. As $\alpha$ increases, the robustness measure enhances, but concurrently, efficiency decreases \citep{das2022testing}. We now turn to the estimation of parameters within the normal distribution, a case of particular relevance to the Mahalanobis distance metric used throughout this article.

\textbf{Univariate case ($y \sim N(u,\sigma^2)$):} Here, the parameter $\btheta = (\mu, \sigma)$ can be estimated through the Minimum Density Power Divergence Estimate (MDPDE) by directly minimizing the DPD measure, as given in Equation \eqref{normal_eqation}. Alternatively, the MDPDE can be obtained through an iterative approach, by repeatedly solving the following estimating equations for the weight $\omega_i$, the mean $\mu$, and the variance $\sigma^2$:
\begin{equation}
\omega_i = \exp \left[-\frac{\alpha}{2}\left(\frac{y_i-\mu}{\sigma}\right)^2\right], \quad
\mu = \frac{\sum_{i=1}^m \omega_i y_i}{\sum_{i=1}^m \omega_i}, \quad
\sigma^2 = \frac{\sum_{i=1}^m \omega_i\left(y_i-u\right)^2}{\sum_{i=1}^m \omega_i-\frac{m\alpha}{(1+\alpha)^{3 / 2}}}.
\label{normal_equation2}
\end{equation}
When $\alpha=0$, Equation \eqref{normal_equation2} reduces exactly to the normal equations for the non-robust ordinary least squares estimates, recovering the classical MLE as a special case.

\textbf{Multivariate case ($y \sim N_p(u,\Sigma)$):}
\label{emper_mult}
Analogous to the univariate case described above, we now consider the multivariate setting, which is directly relevant to real-world clustering applications involving multiple correlated variables. In this case, the MDPDE can be obtained by minimizing
\begin{equation}
V_{\alpha, \mu, \Sigma}({\bf y})=(2 \pi)^{-\frac{p\alpha}{2}}|\Sigma|^{-\alpha / 2}(1+\alpha)^{-\frac{p}{2}}\left[1-\frac{(1+\alpha)^{\frac{p+2}{2}}}{m \alpha} \sum_{i=1}^m \exp \left[-\frac{\alpha}{2} B_i\right]\right]
\label{normal_equation4}
\end{equation}
with respect to $\mu$ and $\Sigma$, where $B_i=\left(y_i-\mu\right)^T \Sigma^{-1}\left(y_i-\mu\right)$ denotes the squared Mahalanobis distance of the $i$-th observation from the mean.

As in the univariate case, the MDPDE can alternatively be obtained by solving the following equations iteratively for the mean $u$, the weight $\omega_i$, and the inverse covariance matrix $\Sigma^{-1}$:
\begin{equation}
    u = \frac{\sum_{i=1}^m \omega_i y_i}{\sum_{i=1}^m \omega_i}, \quad
    \omega_i = \exp \left[-\frac{\alpha}{2} B_i\right], \quad
    \Sigma^{-1} = \left(\sum_{i=1}^m \omega_i-\frac{m \alpha}{(1+\alpha)^{\frac{p+2}{2}}}\right)\left(\sum_{i=1}^m \omega_i\left(y_i-\mu \right)\left(y_i-\mu \right)^{\top}\right)^{-1}.    
\end{equation}
This iterative method requires initial values for $\mu$ and $\Sigma$ to be specified in advance. To improve robustness against outliers, we use the scaled median absolute deviation (MAD) method to obtain these initial estimates of $\mu$ and $\Sigma$.

\subsection{Robust K means with Mahalanobis distance and DPD measure }
\label{sec:mkdpd}

We now introduce a robust version of the K-means algorithm with Mahalanobis distance, in which the maximum likelihood method of estimating the cluster center and the cluster variance-covariance matrix $(\mu_k, \Sigma_k)$ in Step 5 of Algorithm \ref{MK-means} is replaced by a robust method that instead minimizes the DPD distance $V_{\mu_k, \Sigma_k}(\textbf{x}_k)$, where $\textbf{x}_k = (x_{k1}, x_{k2}, \ldots, x_{kn_k})$ is the set of points assigned to cluster $k$ and $n_k$ is the number of points in cluster $k$, to estimate the cluster centers $\mu_k$ and cluster variance-covariance matrices $\Sigma_k$, $k=1,2,\ldots,K$.

The proposed clustering algorithm, from here on referred to as MK-means DPD, is given below in Algorithm \ref{MDPDK}.	

\begin{algorithm}[H]
\caption{K-means using Mahalanobis distance and DPD estimators (MK-means DPD)}
\label{MDPDK}
\ben
\item Set an initial values of the cluster centers $(\hat{\mu}_{1}, \cdots, \hat{\mu}_{K})$. 
\item Assign $\x_1$ to the cluster $k$ if the euclidean distance between $\x_1$ and $\hat{\mu}_{k}$ is the smallest. Repeat this for $\x_2, \x_3, \cdots \x_n$. Let $C_{0k}$ denote the indices of the set of points which are assigned to the $k$-th cluster, $k=1,2, \cdots, K$. 
\item Estimate $\hat{\Sigma}_{k} = \frac{1}{n} \sum_{i\in C_{0k}}(\x_i-\hat{\mu}_{k})(\x_i-\hat{\mu}_{k})^T$, $k=1,2, \cdots, K$.
\item Calculate the distance of $\x_i$ with the center of the cluster $(\hat{\mu}_{k}$ by the Mahalanobis distance given by $d(\x_i, \hat{\mu}_{k})=(\x_i-\hat{\mu}_{k})\hat{\Sigma}_{k}^{-1}(\x_i-\hat{\mu}_{k})^T$, $k=1,2, \cdots, K$. Assign $\x_i$ to cluster $m$ if $min_{k}(d(\x_i, \hat{\mu}_{k}))=d(\x_i, \hat{\mu}_{m}))$. Repeat this for $i=1,2,\cdots,n$ and obtain the new clusters. Let $C_{k}$ denote the set of points assigned to the $k$-th cluster, $k=1,2, \cdots, K$.
\item Update $(\hat{\mu}_{k}, \hat{\Sigma}_k^{-1})$ by the robust DPD estimator using the iterative weighted method described in section \ref{emper_mult}
\item Repeat Steps 4-5 until the changes in $(\hat\mu_k, \hat\Sigma_k)$ between successive iterations fall below a pre-specified tolerance, or until a pre-specified maximum number of iterations is reached.
\een
\end{algorithm}

\subsubsection{Choosing initial cluster centers and covariance matrices}
\label{sec:init}
The algorithm is particularly sensitive to the initial choice of cluster centers and covariance matrices for the clusters. Therefore, it is crucial to have accurate initial estimates for the cluster center. A common approach to mitigate this issue is to run the algorithm multiple times (e.g., N times), randomly selecting initial cluster centers from the data range. This strategy is a routine practice in the traditional K-means algorithm. However, this approach encounters computational complexity challenges, especially for high-dimensional data where N needs to be very large. Furthermore, our algorithm employs the Mahalanobis distance, making the selection of an appropriate initial covariance matrix even more challenging. To address this issue, we adopt a similar strategy as described in \cite{Melnykov}. At the outset, we identify a set of points with a dense concentration of neighbors, characterizing the "core" of the selected cluster. These points provide an initial estimate of the covariance matrix using the DPD measure. Subsequently, this estimate is employed to compute Mahalanobis distances for the remaining points. The anticipation is that the closest points should pertain to the same cluster, and the initial surge in distances suggests the inclusion of points from a distinct cluster. This iterative strategy can be replicated K times, with each iteration eliminating the chosen points to prevent the recurrent selection of the same clusters. The iterative approach for constructing initial values for $\mu$ and $\Sigma$ is outlined below. These steps will substitute Steps 1 and 2 in Algorithm \ref{MDPDK}.

\ben
\item Compute the Euclidean distance $d_{i,j}$ for each pair of data points $x_i$ and $x_j$, $i\neq j = 1, 2, \ldots, n$. Then, for every data point $X_i$, compute the sum of the $w$ smallest distances to it $S_{iw}=\sum_{j=1}^{w}d_{i,(j)}$, where $d_{i,(j)}'s$ are the ordered distances. An approximate value of $w$ is taken to be 20.
Set k=1 and the number of points not assigned to any cluster M=n.
\item Randomly assign one of the unassigned points as the cluster center with probabilities $\frac{1}{i_s}(\sum_{l=1}^{M}\frac{1}{i_l})^-1$, $s=1, 2, \ldots, M$, where $i_1, i_2, \ldots, i_M$ are such that $S_{i_1,w}\leq S_{i_2,w}, \leq \ldots \leq S_{i_M,w}$. 
\item Let D be the minimum number of points in a cluster, assign D points nearest to the center of the current cluster $C_k$. We used D = 20, in our simulations.
\item Estimate $\mu_k$ and $\Sigma_k$ by the robust DPD estimator using the iterative weighted method described in section \ref{emper_mult} with the data points assigned to the cluster.
\item Using the $99 \%$ confidence ellipsoid, include all points in the k-th cluster satisfying $(x_i -\mu_k)^T\Sigma_k^{-1}(x_i-\mu_k)\leq \chi^2_{p,0.01}$, where $p$ is the dimension of $x$.
\item Repeat steps 3,4 and 5 a few times (usually no more than 5 times) to obtain a initial estimate of $\mu_k$ and $\Sigma_k$
\item Set k=k+1 and go to step 2 and repeat till $k=K$ or the number of unassigned points reaches the minimum value D. 
\een

This method of estimating initial values for $\mu_k$ and $\Sigma_k$ can also guide us in identifying the optimal number of clusters for a given dataset. For instance, if the actual number of clusters is less than $K$ (the predefined number of clusters), it is anticipated that Step 7 would halt before reaching $K$ clusters, as there would be no more points left to form a new cluster.

\begin{remark}[Advantages of MK-means DPD]
The robustness of MK-means DPD to outliers stems directly from the density power divergence-based estimation of the cluster centers and covariance matrices, which allows the algorithm to produce reliable estimates even when the data contain unusual or erroneous values. This robustness is especially valuable in high-dimensional settings, where outliers and extreme values are more prevalent and conventional distance-based methods are particularly susceptible to their influence. Beyond robust estimation, the algorithm also facilitates the direct detection of outliers, through the use of the estimated cluster centers together with the corresponding confidence ellipsoids, as illustrated in the simulation study of Section 4; notably, this is achieved without removing any observations from the dataset during the clustering process, so no reduction in the available data occurs. A further advantage lies in the flexibility of the tuning parameter $\alpha$, which allows the degree of sensitivity to deviations from the assumed probability model to be adjusted according to the level of contamination present in the data, while the use of the Mahalanobis distance further allows the method to accommodate non-spherical and elliptical cluster shapes. Taken together, these properties make the DPD-based estimation approach well suited to data exhibiting intricate distributions and non-linear relationships among variables.
\end{remark}

\begin{remark}[Drawbacks of MK-means DPD]
Despite these advantages, MK-means DPD has some limitations. The optimization involved in the DPD-based estimation step can be computationally demanding, particularly for large or high-dimensional datasets, since the cluster centers and covariance matrices must be re-estimated iteratively at every step of the algorithm. Selecting an appropriate value for the tuning parameter $\alpha$ can also pose a practical challenge: an inappropriate choice may lead to biased or inefficient estimates, and while $\alpha$ could in principle be selected automatically via cross-validation, doing so would add considerable computational cost. Finally, like many DPD-based methods, MK-means DPD relies on a Gaussian distributional assumption for each cluster; when the true underlying distribution departs substantially from normality, the performance of the algorithm may deteriorate, resulting in biased or inaccurate cluster estimates. As discussed in Section \ref{sec:dc_mkmeans}, MK-means DPD also inherits the convergence limitations of the Mahalanobis K-means algorithm on which it is built, a limitation addressed by the convergent variant introduced there.
\end{remark}

\subsection{A Convergent Variant: Density-Consistent MK-means DPD (DC-MK-means DPD)}
\label{sec:dc_mkmeans}

A limitation of Algorithm \ref{MDPDK} that we have not yet addressed is that it does not come with any convergence guarantee. This is not a limitation unique to the DPD-based robustification introduced in this article; rather, it is inherited directly from the Mahalanobis K-means algorithm (Algorithm \ref{MK-means}) on which Algorithm \ref{MDPDK} is built. \citet{lapidot2018convergence} showed, both analytically and empirically, that the assignment step and the update step of Mahalanobis K-means target two different objectives, and that this mismatch can prevent the overall distance objective from decreasing monotonically across iterations. Writing $\Lambda_k := \hat\Sigma_k^{-1}$ for the precision matrix of cluster $k$, following the notation of \citet{lapidot2018convergence}, and $D_{Mahal} = \sum_{k=1}^K\sum_{i\in C_k} (\x_i-\hat\mu_k)^T\Lambda_k(\x_i-\hat\mu_k)$ for the assignment objective minimized in Step 4 of Algorithm \ref{MK-means}, \citet{lapidot2018convergence} showed that $\partial D_{Mahal}/\partial\Lambda_k$ does not depend on $\Lambda_k$ at all, so that $D_{Mahal}$ is minimized exactly, and trivially, by setting $\Lambda_k=0$ for every cluster (equivalently, $\hat\Sigma_k \to \infty$), a degenerate solution in which every observation is deemed equally close to every cluster, regardless of the data or the assignment. The covariance update actually used in Step 5 of Algorithm \ref{MK-means} and \ref{MDPDK}, by contrast, is not obtained by minimizing $D_{Mahal}$ at all: it substitutes the corresponding maximum likelihood covariance estimator, borrowed from the Gaussian case by analogy, rather than any quantity connected to the assignment objective. In short, the assignment step and the update step target two different objectives, the former trivially minimized by unconstrained covariance inflation and the latter estimating the covariance as though maximizing an unrelated Gaussian likelihood, so there is no single, well-defined objective function that both steps can be shown to jointly decrease, and the monotonicity argument underlying the classical convergence proof of \citet{MacQueen1967} and \citet{lloyd1982least} for Euclidean K-means does not carry over. \citet{lapidot2018convergence} confirmed this behavior empirically: unlike standard K-means and several corrected variants, the Mahalanobis K-means distortion was observed to both increase and decrease unpredictably across iterations. Notably, \citet{lapidot2018convergence} also showed that this problem can be repaired by replacing the assignment criterion with the negative Gaussian log-likelihood, so that both the assignment and update steps act on the same underlying objective; it is precisely this strategy, of aligning the assignment criterion with the objective minimized in the update step, that we generalize below to the DPD setting, and that underlies the convergence guarantee established later in this section.

We now introduce a variant of Algorithm \ref{MDPDK}, which we call \emph{Density-Consistent MK-means DPD} (DC-MK-means DPD), that resolves this issue by replacing the assignment criterion in Step 4 with a quantity that is directly and provably consistent with the objective minimized in the update step, thereby restoring a Lloyd/MacQueen-style convergence guarantee, in the same spirit as the Gaussian ML K-means variant of \citet{lapidot2018convergence} discussed above, but adapted to the DPD-based objective used throughout this article. The key observation is that the multivariate DPD objective $V_{\alpha,\mu,\Sigma}(\x_k)$ used in the update step, given in Equation \eqref{normal_equation4}, is, by construction, an average of per-observation terms, exactly as in the original formulation of the DPD objective of \citet{basu1998robust} given in Equation \eqref{normal_eqation}. Writing $C_\alpha(\Sigma) := (2\pi)^{-p\alpha/2}|\Sigma|^{-\alpha/2}(1+\alpha)^{-p/2}$ and distributing this constant across the sum in Equation \eqref{normal_equation4}, we obtain
\begin{equation}
V_{\alpha,\mu,\Sigma}(\x_k) = \frac{1}{m}\sum_{i=1}^m v_\alpha(x_{ki};\mu,\Sigma), \label{eq:pointwise}
\end{equation}
where the pointwise DPD loss $v_\alpha(x_{ki};\mu,\Sigma)$ is given by
\begin{equation}
v_\alpha(x_{ki};\mu,\Sigma) := C_\alpha(\Sigma)\left[1-\frac{(1+\alpha)^{\frac{p+2}{2}}}{\alpha}\exp\left(-\frac{\alpha}{2}B_{ki}\right)\right], \qquad B_{ki} = (x_{ki}-\mu)^T\Sigma^{-1}(x_{ki}-\mu). \label{eq:vloss}
\end{equation}

This structural fix directly mirrors the repair validated in Section II-C of \citet{lapidot2018convergence}: just as replacing $D_{Mahal}$ with the Gaussian negative log-likelihood $D_{Gauss}$ introduces a $-\log|\Lambda_k|$ term that makes $\partial D_{Gauss}/\partial \Lambda_k$ properly depend on $\Lambda_k$, in contrast to $\partial D_{Mahal}/\partial\Lambda_k$, from which $\Lambda_k$ drops out entirely, the multivariate DPD objective $V_{\alpha,\mu,\Sigma}$ carries the analogous determinant-dependent factor $C_\alpha(\Sigma) \propto |\Sigma|^{-\alpha/2}$, so that its stationarity condition with respect to $\Sigma^{-1}$ is a genuine equation in $\Sigma^{-1}$ rather than one from which $\Sigma^{-1}$ has vanished. Indeed, setting $\partial V_{\alpha,\mu,\Sigma}(\x_k)/\partial \Sigma^{-1} = 0$ yields exactly the estimating equation
\begin{equation}
\Sigma^{-1} = \left(\sum_{i=1}^m \omega_i - \frac{m\alpha}{(1+\alpha)^{\frac{p+2}{2}}}\right)\left(\sum_{i=1}^m \omega_i(\x_i-\mu)(\x_i-\mu)^\top\right)^{-1} \label{eq:sigmastationarity}
\end{equation}
used in Step 5 of Algorithm \ref{MDPDK} (Section \ref{emper_mult}). Consequently, unlike the heuristic covariance substitution used in the Mahalanobis case, the update step of Algorithm \ref{MDPDK} is already a genuine stationary point of the same objective $V_{\alpha,\mu,\Sigma}$ that we now also use for assignment. Because $v_\alpha$ is exactly the pointwise term whose average defines $V_{\alpha,\mu,\Sigma}$ (Equation \eqref{eq:pointwise}), using $v_\alpha$ for assignment as well as for estimation ensures that both steps of the algorithm act on the same underlying objective, resolving the assignment/update mismatch that broke the monotonicity argument for Algorithm \ref{MDPDK}. DC-MK-means DPD is identical to Algorithm \ref{MDPDK} in every respect except that Step 4 assigns each observation to the cluster minimizing $v_\alpha(\x_i;\hat\mu_k,\hat\Sigma_k)$ in place of $B_i$, as detailed below.

\begin{algorithm}[H]
\caption{Density-Consistent MK-means DPD (DC-MK-means DPD)}
\label{DCMDPDK}
\ben
\item Set an initial values of the cluster centers $(\hat{\mu}_{1}, \cdots, \hat{\mu}_{K})$. 
\item Assign $\x_1$ to the cluster $k$ if the euclidean distance between $\x_1$ and $\hat{\mu}_{k}$ is the smallest. Repeat this for $\x_2, \x_3, \cdots \x_n$. Let $C_{0k}$ denote the indices of the set of points which are assigned to the $k$-th cluster, $k=1,2, \cdots, K$. 
\item Estimate $\hat{\Sigma}_{k} = \frac{1}{n} \sum_{i\in C_{0k}}(\x_i-\hat{\mu}_{k})(\x_i-\hat{\mu}_{k})^T$, $k=1,2, \cdots, K$.
\item Calculate the pointwise DPD loss of $\x_i$ with respect to cluster $k$, $v_\alpha(\x_i;\hat\mu_k,\hat\Sigma_k)$, as given in Equation \eqref{eq:vloss}, $k=1,2, \cdots, K$. Assign $\x_i$ to cluster $m$ if $\min_{k}\left(v_\alpha(\x_i;\hat\mu_k,\hat\Sigma_k)\right)=v_\alpha(\x_i;\hat\mu_m,\hat\Sigma_m)$. Repeat this for $i=1,2,\cdots,n$ and obtain the new clusters. Let $C_{k}$ denote the set of points assigned to the $k$-th cluster, $k=1,2, \cdots, K$.
\item Update $(\hat{\mu}_{k}, \hat{\Sigma}_k^{-1})$ by first applying the iterative weighted method described in Section \ref{emper_mult}, warm-started at the current estimate. If the resulting update does not decrease $V_{\alpha,\hat\mu_k,\hat\Sigma_k}(\x_k)$ relative to the current estimate, obtain $(\hat{\mu}_{k}, \hat{\Sigma}_k^{-1})$ instead via a globally convergent Newton-type step with Armijo backtracking, warm-started at the current estimate, as detailed in the remark below.
\item Repeat Steps 4-5 until the changes in $(\hat\mu_k, \hat\Sigma_k)$ between successive iterations fall below a pre-specified tolerance.
\een
\end{algorithm}

\begin{remark}[Ensuring the monotonicity requirement of Step 5]
In the typical case, Step 5 is carried out directly via the iterative weighted method of Section \ref{emper_mult}, warm-started at the previous parameter estimate. In the event that this scheme, so warm-started, fails to return a value of $V_{\alpha,\mu,\Sigma}(\x_k)$ at least as small as at the starting point, Step 5 instead applies a globally convergent Newton-type procedure with an Armijo backtracking line search, warm-started at the same point, following the same general strategy as the univariate hybrid Newton/gradient-descent procedure of \citet{anum2024hybrid} for density power divergence estimation, extended here to the multivariate parameter space $(\mu,\Sigma)$; this procedure is guaranteed to satisfy the required non-increase property. Writing $\Lambda = \Sigma^{-1}$, $\omega_i = \exp(-\frac{\alpha}{2}B_i)$, $B_i = (\x_i-\mu)^T\Lambda(\x_i-\mu)$, and $K = (2\pi)^{-p\alpha/2}(1+\alpha)^{-p/2}$, $A = (1+\alpha)^{(p+2)/2}/\alpha$, the gradients of $V_{\alpha,\mu,\Sigma}(\x_k) = K|\Lambda|^{\alpha/2}\left[1-\frac{A}{m}\sum_{i=1}^m\omega_i\right]$ needed for this procedure are
\begin{align}
\nabla_\mu V &= -\frac{KA\alpha}{m}|\Lambda|^{\alpha/2}\,\Lambda\sum_{i=1}^m \omega_i(\x_i-\mu), \\
\nabla_\Lambda V &= \frac{K\alpha}{2}|\Lambda|^{\alpha/2}\left[\Sigma\left(1-\frac{A}{m}\sum_{i=1}^m\omega_i\right) + \frac{A}{m}\sum_{i=1}^m\omega_i(\x_i-\mu)(\x_i-\mu)^\top\right],
\end{align}
which recover the estimating equations of Section \ref{emper_mult} upon setting $\nabla_\mu V = \nabla_\Lambda V = 0$. Writing $\theta = (\mu,\operatorname{vech}(\Lambda))$ for the stacked parameter vector, where $\operatorname{vech}(\Lambda)$ collects the distinct entries of the symmetric matrix $\Lambda$, a quasi-Newton (BFGS) direction $d^{(j)} = -H_j^{-1}\nabla V(\theta^{(j)})$ is computed at each inner iteration $j$, where $H_j$ is a positive-definite approximation of the Hessian of $V$ built from the gradients above. The step size $t$ is then chosen by Armijo backtracking: starting from $t=1$, $t$ is accepted if
\begin{equation}
V(\theta^{(j)} + t\,d^{(j)}) \le V(\theta^{(j)}) + c_1\, t\, \nabla V(\theta^{(j)})^\top d^{(j)},
\end{equation}
for a fixed constant $c_1 \in (0,1)$, and halved otherwise, until the condition is satisfied. Since $H_j$ is positive definite, $d^{(j)}$ is a genuine descent direction, so this backtracking loop terminates after finitely many halvings, and each accepted step satisfies $V(\theta^{(j+1)}) \le V(\theta^{(j)})$; iterating this procedure therefore yields a final value of $V_{\alpha,\mu,\Sigma}(\x_k)$ no larger than at the warm-started starting point, satisfying the requirement used in the proof of Proposition \ref{prop:monotone}.
\end{remark}

\begin{remark}[Initialization for DC-MK-means DPD]
As with Algorithm \ref{MDPDK}, the choice of initial cluster centers and covariance matrices in Steps 1 and 2 of Algorithm \ref{DCMDPDK} can be replaced by the iterative initialization scheme described in Section \ref{sec:init}, which provides more accurate initial estimates of $\mu_k$ and $\Sigma_k$ than random initialization and thereby improves the practical performance of the algorithm.
\end{remark}

We now show that, unlike Algorithm \ref{MDPDK}, Algorithm \ref{DCMDPDK} admits a well-defined joint objective function that is non-increasing at every iteration. Define, for a partition $\mathcal{C}=(C_1,\ldots,C_K)$ and parameters $\{\hat\mu_k,\hat\Sigma_k\}_{k=1}^K$,
\begin{equation}
Q\left(\mathcal{C},\{\hat\mu_k,\hat\Sigma_k\}\right) := \sum_{k=1}^K \sum_{i \in C_k} v_\alpha\left(\x_i;\hat\mu_k,\hat\Sigma_k\right) = \sum_{k=1}^K n_k\, V_{\alpha,\hat\mu_k,\hat\Sigma_k}(\x_k). \label{eq:Qdef}
\end{equation}

\begin{proposition}[Monotonic non-increase]
\label{prop:monotone}
Let $\mathcal{C}^{(t)}$ and $\{\hat\mu_k^{(t)},\hat\Sigma_k^{(t)}\}$ denote, respectively, the partition and parameter estimates produced after $t$ iterations of Steps 4--5 of Algorithm \ref{DCMDPDK}, and suppose that at each iteration the update in Step 5 is run until the DPD estimating equations (Section \ref{emper_mult}) are satisfied for every non-empty cluster. Then
\begin{equation}
Q\left(\mathcal{C}^{(t+1)}, \{\hat\mu_k^{(t+1)},\hat\Sigma_k^{(t+1)}\}\right) \le Q\left(\mathcal{C}^{(t)}, \{\hat\mu_k^{(t)},\hat\Sigma_k^{(t)}\}\right), \qquad t = 0,1,2,\ldots
\end{equation}
\end{proposition}

\begin{proof}
We show the two steps of iteration $t+1$ separately.

\emph{(i) The assignment step does not increase $Q$.} Fix $\{\hat\mu_k^{(t)},\hat\Sigma_k^{(t)}\}$. By Equation \eqref{eq:Qdef}, $Q(\mathcal C, \{\hat\mu_k^{(t)},\hat\Sigma_k^{(t)}\}) = \sum_{i=1}^n v_\alpha(\x_i;\hat\mu_{c(i)}^{(t)},\hat\Sigma_{c(i)}^{(t)})$ for any partition $\mathcal C$, where $c(i)$ denotes the cluster to which $\mathcal C$ assigns observation $i$. Step 4 of Algorithm \ref{DCMDPDK} forms $\mathcal{C}^{(t+1)}$ by assigning each $\x_i$ independently to $\arg\min_k v_\alpha(\x_i;\hat\mu_k^{(t)},\hat\Sigma_k^{(t)})$. Since each term of the sum is minimized independently and the terms do not interact across observations, no alternative partition can achieve a smaller value of this sum than $\mathcal C^{(t+1)}$, in particular not $\mathcal C^{(t)}$. Hence $Q(\mathcal C^{(t+1)}, \{\hat\mu_k^{(t)},\hat\Sigma_k^{(t)}\}) \le Q(\mathcal C^{(t)}, \{\hat\mu_k^{(t)},\hat\Sigma_k^{(t)}\})$.

\emph{(ii) The update step does not increase $Q$.} Fix $\mathcal{C}^{(t+1)}$. By construction, Step 5 obtains $(\hat\mu_k^{(t+1)},\hat\Sigma_k^{(t+1)})$ by solving the DPD estimating equations of Section \ref{emper_mult} for the data in cluster $C_k^{(t+1)}$, for each $k$ separately. As shown above (Equation \eqref{eq:sigmastationarity}), these estimating equations are precisely the first-order (stationarity) conditions for minimizing $V_{\alpha,\mu,\Sigma}(\x_k^{(t+1)})$ over $(\mu,\Sigma)$. We require that Step 5 be carried out by a numerical procedure, warm-started at the previous iterate $(\hat\mu_k^{(t)},\hat\Sigma_k^{(t)})$ — typically the iterative weighted method of Section \ref{emper_mult} itself, falling back where necessary to the safeguarded Newton-type procedure detailed in the remark following Algorithm \ref{DCMDPDK} — such that the resulting solution $(\hat\mu_k^{(t+1)},\hat\Sigma_k^{(t+1)})$ satisfies $V_{\alpha,\hat\mu_k^{(t+1)},\hat\Sigma_k^{(t+1)}}(\x_k^{(t+1)}) \le V_{\alpha,\hat\mu_k^{(t)},\hat\Sigma_k^{(t)}}(\x_k^{(t+1)})$ for every $k$ with $C_k^{(t+1)}\neq\emptyset$. Multiplying by $n_k^{(t+1)}$ and summing over $k$ gives $Q(\mathcal{C}^{(t+1)}, \{\hat\mu_k^{(t+1)},\hat\Sigma_k^{(t+1)}\}) \le Q(\mathcal{C}^{(t+1)}, \{\hat\mu_k^{(t)},\hat\Sigma_k^{(t)}\})$.

Combining (i) and (ii), $Q(\mathcal{C}^{(t+1)}, \{\hat\mu_k^{(t+1)},\hat\Sigma_k^{(t+1)}\}) \le Q(\mathcal{C}^{(t+1)}, \{\hat\mu_k^{(t)},\hat\Sigma_k^{(t)}\}) \le Q(\mathcal{C}^{(t)}, \{\hat\mu_k^{(t)},\hat\Sigma_k^{(t)}\})$.
\end{proof}

To turn Proposition \ref{prop:monotone} into a finite convergence guarantee, we impose one further, standard regularity condition: there exists $\delta>0$ such that every $\hat\Sigma_k^{(t)}$ produced by the algorithm has all eigenvalues bounded below by $\delta$. This is the same type of non-degeneracy condition routinely imposed in the finite mixture model literature to rule out the well-known unbounded-likelihood problem that arises when a component's covariance is allowed to collapse around a single point.

\begin{thm}[Finite convergence of DC-MK-means DPD]
\label{thm:convergence}
Suppose that, throughout the run of Algorithm \ref{DCMDPDK}, the eigenvalues of every $\hat\Sigma_k^{(t)}$ remain bounded below by some $\delta>0$, and that ties in the assignment step of Step 4 are broken by a fixed, consistent rule (e.g., retaining the previous assignment). Then the sequence of objective values $Q^{(t)} := Q(\mathcal{C}^{(t)},\{\hat\mu_k^{(t)},\hat\Sigma_k^{(t)}\})$ is non-increasing and bounded below, and hence converges to a finite limit. Moreover, the sequence of partitions $\mathcal{C}^{(t)}$ stabilizes after finitely many iterations, so that Algorithm \ref{DCMDPDK} converges in a finite number of steps to a partition $\mathcal C^*$ and parameters $\{\hat\mu_k^*,\hat\Sigma_k^*\}$ that are simultaneously stable under Step 4 and Step 5.
\end{thm}

\begin{proof}
\emph{Boundedness below.} Since $\sum_{i=1}^m \exp(-\frac{\alpha}{2}B_{ki}) \in [0,m]$ for any $\mu,\Sigma$, the bracketed term in Equation \eqref{normal_equation4} satisfies
\begin{equation}
1-\frac{(1+\alpha)^{\frac{p+2}{2}}}{\alpha} \;\le\; 1-\frac{(1+\alpha)^{\frac{p+2}{2}}}{m\alpha}\sum_{i=1}^m\exp\left(-\frac{\alpha}{2}B_{ki}\right) \;\le\; 1,
\end{equation}
and since $C_\alpha(\Sigma) > 0$, multiplying through by $C_\alpha(\Sigma)$ gives
\begin{equation}
V_{\alpha,\mu,\Sigma}(\x_k) \;\ge\; C_\alpha(\Sigma)\left(1-\frac{(1+\alpha)^{\frac{p+2}{2}}}{\alpha}\right). \label{eq:Vlowerbound}
\end{equation}
Under the eigenvalue lower bound $\delta$, $|\hat\Sigma_k^{(t)}|\ge \delta^p$, so
\begin{equation}
C_\alpha(\hat\Sigma_k^{(t)}) = (2\pi)^{-p\alpha/2}|\hat\Sigma_k^{(t)}|^{-\alpha/2}(1+\alpha)^{-p/2} \;\le\; (2\pi)^{-p\alpha/2}\delta^{-p\alpha/2}(1+\alpha)^{-p/2} =: M < \infty. \label{eq:Cbound}
\end{equation}
Writing $c := 1-\frac{(1+\alpha)^{\frac{p+2}{2}}}{\alpha}$, which is negative for typical $\alpha \in (0,1)$, multiplying \eqref{eq:Cbound} by the negative constant $c$ reverses the inequality direction, giving $C_\alpha(\hat\Sigma_k^{(t)})\, c \ge Mc$; combined with \eqref{eq:Vlowerbound}, this yields $V_{\alpha,\hat\mu_k^{(t)},\hat\Sigma_k^{(t)}}(\x_k^{(t)}) \ge Mc = -M\left(\frac{(1+\alpha)^{\frac{p+2}{2}}}{\alpha}-1\right) =: -L$ for every $k$ and $t$, and therefore
\begin{equation}
Q^{(t)} = \sum_k n_k^{(t)} V_{\alpha,\hat\mu_k^{(t)},\hat\Sigma_k^{(t)}}(\x_k^{(t)}) \;\ge\; -nL,
\end{equation}
a finite lower bound not depending on $t$.

\emph{Convergence of the objective value.} By Proposition \ref{prop:monotone}, $\{Q^{(t)}\}$ is non-increasing; combined with the lower bound $-nL$ just established, the monotone convergence theorem for real sequences guarantees that $Q^{(t)} \to Q^*$ for some finite $Q^* \ge -nL$.

\emph{Finite termination of the partition sequence.} There are only finitely many ways to partition $n$ observations into at most $K$ non-empty groups; denote this finite collection of partitions by $\mathfrak{P}$. Under the hypothesis of Proposition \ref{prop:monotone} that Step 5 is run until the DPD estimating equations are satisfied exactly at every iteration, and assuming, as is standard in M-estimation, that these estimating equations admit an essentially unique solution for any fixed, non-empty cluster data, the parameter values $\{\hat\mu_k,\hat\Sigma_k\}$ produced by Step 5 for a given partition $\mathcal C \in \mathfrak P$ depend only on $\mathcal C$ itself, not on the iteration at which that partition is visited. Consequently $Q^{(t)}$ takes values only in the finite set $\{Q(\mathcal C) : \mathcal C \in \mathfrak P\}$, where $Q(\mathcal C)$ denotes the value of $Q$ at the (essentially unique) parameter estimates associated with partition $\mathcal C$.

Whenever $\mathcal{C}^{(t+1)} \ne \mathcal{C}^{(t)}$, the tie-breaking rule assumed above ensures that at least one observation is reassigned to a cluster with a strictly smaller pointwise loss, so combining the strict decrease from the assignment-step argument in the proof of Proposition \ref{prop:monotone}(i) with the non-increase from part (ii) gives
\begin{equation}
Q^{(t+1)} \;\le\; Q\left(\mathcal{C}^{(t+1)}, \{\hat\mu_k^{(t)},\hat\Sigma_k^{(t)}\}\right) \;<\; Q\left(\mathcal{C}^{(t)}, \{\hat\mu_k^{(t)},\hat\Sigma_k^{(t)}\}\right) = Q^{(t)}, \qquad \text{whenever } \mathcal{C}^{(t+1)} \ne \mathcal{C}^{(t)}. \label{eq:strictdecrease}
\end{equation}
Since $\{Q^{(t)}\}$ is a non-increasing sequence taking values in the finite set $\{Q(\mathcal C):\mathcal C \in \mathfrak P\}$, Equation \eqref{eq:strictdecrease} can hold for at most $|\mathfrak P|-1$ values of $t$ before $Q^{(t)}$ must remain constant; as any further change in $\mathcal{C}^{(t)}$ would force a further strict decrease, the partition sequence $\mathcal{C}^{(t)}$ must therefore also stabilize, at some finite iteration $T \le |\mathfrak{P}|-1$, i.e., $\mathcal{C}^{(t)}=\mathcal{C}^{(T)}=:\mathcal{C}^*$ for all $t \ge T$.

\emph{Stability of the fixed point.} Once the partition no longer changes, Step 4 no longer reassigns any observation, so $\mathcal C^*$ is stable under the assignment step given $\{\hat\mu_k^{(T)},\hat\Sigma_k^{(T)}\}$; and since Step 5 solves the DPD estimating equations exactly for $\mathcal{C}^{(T)}=\mathcal{C}^*$, the resulting $\{\hat\mu_k^*,\hat\Sigma_k^*\} := \{\hat\mu_k^{(T)},\hat\Sigma_k^{(T)}\}$ is also stable under Step 5. Hence the algorithm has converged, in a finite number of iterations, to a partition and parameter set that is simultaneously stable under both steps.
\end{proof}

\begin{remark}
Theorem \ref{thm:convergence} establishes that the partition sequence $\mathcal{C}^{(t)}$ stabilizes after finitely many iterations, i.e., $\mathcal{C}^{(t+1)}=\mathcal{C}^{(t)}$ for all $t$ beyond some finite $T$. Once this occurs, Step 5 re-solves the same DPD estimating equations (Section \ref{emper_mult}) for the same, unchanged cluster data $C_k^{(T)}$; provided this fixed-data estimation problem admits an essentially unique solution, a standard regularity condition in M-estimation, the resulting parameter estimates $(\hat\mu_k^{(t)},\hat\Sigma_k^{(t)})$ also cease to change for $t \geq T$. Consequently, checking whether $(\hat\mu_k,\hat\Sigma_k)$ has stabilized between successive iterations, the same practical criterion used for Algorithms \ref{MK-means} and \ref{MDPDK}, is, for Algorithm \ref{DCMDPDK}, a theoretically justified proxy for partition stabilization: unlike for Algorithms \ref{MK-means} and \ref{MDPDK}, where no such guarantee exists, Theorem \ref{thm:convergence} ensures that this criterion will indeed be met after finitely many iterations.
\end{remark}

\begin{remark}
As with classical K-means and its Bregman-divergence generalizations \citep{selim1984kmeans}, Theorem \ref{thm:convergence} guarantees convergence to a stationary point of $Q$, not necessarily its global minimum, since $V_{\alpha,\mu,\Sigma}$ is generally non-convex in $(\mu,\Sigma)$ for $\alpha>0$ and the assignment step is a combinatorial, rather than continuous, optimization. As is standard practice for K-means-type algorithms, multiple random initializations, following the scheme described in Section \ref{sec:init}, help mitigate the risk of convergence to a poor local optimum in this setting as well.
\end{remark}

\begin{remark}
DC-MK-means DPD requires no additional inputs beyond those already used by Algorithm \ref{MDPDK}; the only change is that the assignment step in Step 4 additionally evaluates $|\hat\Sigma_k|$ for each cluster, a quantity already computed as part of the covariance update in Step 5, so the extra computational cost of enforcing convergence is negligible in practice.
\end{remark}

\section{Simulation results and comparison to other clustering algorithms}
In this section, we evaluate the performance of the proposed MK-means DPD and DC-MK-means DPD algorithms using simulated data, comparing them against standard K-means, Mahalanobis K-means, DBSCAN, and a Gaussian Mixture Model. We first introduce the internal and external evaluation metrics used throughout this article, including robust counterparts to several classical internal metrics that are themselves sensitive to outliers, before applying these metrics to a simulated dataset containing known cluster memberships and artificially introduced outliers.

\subsection{Performance Evaluation Metrics}
\label{sec:meth2}

In this subsection, we introduce the performance metrics used to assess the proposed methods against existing clustering methods, categorized into internal evaluation metrics, which assess clustering quality using only the data and the resulting partition, and external evaluation metrics, which compare the clustering result against known class labels when available. Because several classical internal metrics are themselves sensitive to outliers, an issue directly relevant to the robustness motivation of this article, we also introduce robust counterparts to these metrics before turning to the external evaluation metrics.

\subsubsection{Internal Evaluation Metrics}
Internal evaluation metrics assess the quality of clustering outcomes without depending on external information such as ground truth labels, relying solely on the geometry of the data and the partition produced by the clustering algorithm. Because no true class labels are required, internal metrics can be computed for any clustering result, including real-world applications where the correct groupings are unknown. Below, we describe the internal evaluation metrics used in this article.

\textbf{$R^2$:} $R^2$ is an internal evaluation metric that compares SSE (indicating dissimilarity between groups or clusters) to SST (representing total dissimilarity in the dataset). It is measured on a scale from 0 to 1, where a higher $R^2$ suggests well-separated and internally homogeneous clusters \citep{sharma1995applied}. It is defined as:
\begin{equation}
R^2 = 1 - \frac{SSE}{SST}. \label{eq:R2}
\end{equation}

\textbf{Silhouette Index (SI):} The silhouette index (SI) is a metric employed in cluster analysis for evaluating clustering quality. It gauges the similarity of an object to its own cluster in comparison to other clusters. For a data point $x_i$, the silhouette value is given by $S(x_i) = \frac{b_i - a_i}{\max \{a_i, b_i\}}$, where $ a_i = \frac{1}{|C_i|-1} \sum_{j \in C_i, j \neq i} d(i,j)$
is the average distance between observation $i$ and all other points in the same cluster $C_i$, and
$b_i = \min_{k \neq C_i} \left( \frac{1}{|C_k|} \sum_{j \in C_k} d(i,j) \right)$
is the minimum average distance between observation $i$ and all points in any other cluster $C_k$. The overall silhouette index is the average of $S(x_i)$ over all observations:
\begin{equation}
  SI = \frac{1}{n}\sum_{i=1}^n S\left(x_i\right). \label{eq:SI}
\end{equation}
The silhouette index ranges from [-1, +1], where a higher value indicates a better clustering algorithm. However, a challenge with the silhouette index is that it tends to be generally higher for convex clusters compared to other cluster concepts \citep{jumadi2019enhancement}.

\textbf{Davies-Bouldin Index (DBI):} The Davies-Bouldin Index (DBI) stands as an internal evaluation metric assessing the average similarity between clusters and their most similar counterparts. It is bounded below by zero, with a lower score indicating superior clustering performance. DBI is particularly effective for convex clusters, as seen in DBSCAN. Writing $\sigma_i$ for the average distance of all elements in cluster $i$ to centroid $c_i$, and $K$ for the number of clusters, its calculation is as follows:
\begin{equation}
    D B=\frac{1}{K} \sum_{i=1}^K \max _{j \neq i}\left(\frac{\sigma_i+\sigma_j}{d\left(c_i, c_j\right)}\right). \label{eq:DBI}
\end{equation}

\textbf{Calinski-Harabasz Index (CHI):} The Calinski-Harabasz Index (CHI) gauges the ratio between dispersion within clusters and dispersion between clusters. A higher score signifies superior clustering outcomes. The Calinski-Harabasz index typically performs better on convex clusters. Writing $\operatorname{tr}\left(B_K\right)$ for the trace of the dispersion matrix between groups and $\operatorname{tr}\left(W_K\right)$ for the trace of the dispersion matrix within the group \citep{liu2022new}, it is expressed mathematically as:
\begin{equation}
 CH=\frac{\operatorname{tr}\left(B_K\right)}{\operatorname{tr}\left(W_K\right)} \times \frac{n-K}{K-1}.  \label{eq:CHI}
\end{equation}

\subsubsection{Robust Internal Evaluation Metrics}
\label{sec:robust_metrics}

The internal evaluation metrics defined above, $R^2$, SI, DBI, and CHI, all rely, either directly or indirectly, on averages and squared Euclidean distances computed around cluster means. A single extreme observation can therefore disproportionately influence these metrics, potentially misrepresenting the quality of an otherwise well-formed clustering solution, in much the same way that outliers can distort the cluster centers and covariance estimates used by the clustering algorithm itself. This concern is particularly relevant for the present study, given the robustness motivation underlying the proposed MK-means DPD method. We therefore employ robust versions of these indexes, in place of their classical counterparts, throughout the performance comparisons reported later in this article. PAMSIL and trimmed $R^2$ are established, widely-used robust alternatives to the silhouette index and $R^2$, respectively, already found in the literature. In addition, we introduce two new robust indexes based on DBI and CHI, which we refer to as the Median Davies-Bouldin Index ($DBI_M$) and the Trimmed Calinski-Harabasz Index ($CHI^{trim}$). These robust indexes are defined below.

\textbf{Trimmed $R^2$ ($TR^2$):} $R^2$, as defined in Equation~\eqref{eq:R2}, can be influenced by outliers, and we therefore use the established, widely-used trimmed $R^2$ \citep{cuesta1997trimmed} as a performance metric instead. Trimmed $R^2$ assesses the proportion of variance in the data explained by the clustering algorithm while excluding a specified percentage of extreme observations. Mathematically, it is defined as:
\begin{equation}
TR^2 = 1 - \frac{SSE_{trimmed}}{SST_{trimmed}}. \label{eq:TR2}
\end{equation}

\textbf{PAM (Medoid) Silhouette Index (PAMSIL):} Rather than representing a cluster by its mean, the Partitioning Around Medoids (PAM) algorithm \citep{kaufman2009finding} lets $m_k$ denote the \emph{medoid} of cluster $C_k$, defined as $m_k = \argmin_{x_i \in C_k} \sum_{x_j \in C_k} d(x_i, x_j)$, for $k=1,\ldots,K$. \citet{vanderlaan2003new} proposed optimizing the silhouette directly around such medoids, referring to this approach as PAMSIL, later formalized and made computationally efficient by \citet{lenssen2022clustering}. Writing $a_i^{med} = d\left(x_i, m_{k(i)}\right)$ and $b_i^{med} = \min_{l \neq k(i)} d\left(x_i, m_l\right)$ for the distance of $x_i$ to its own medoid and to the nearest medoid of any other cluster, and $S^{med}(x_i) = \frac{b_i^{med} - a_i^{med}}{\max\{a_i^{med}, b_i^{med}\}}$, the PAM (Medoid) Silhouette Index is given by
\begin{equation}
    PAMSIL = \frac{1}{n} \sum_{i=1}^n S^{med}(x_i). \label{eq:PAMSIL}
\end{equation}
Since medoids are robust to outliers, this index is less sensitive to outliers than the classical silhouette index.

\textbf{Median Davies-Bouldin Index ($DBI_M$):} We define the Median Davies-Bouldin Index by replacing the mean-based cluster center $c_i$ with a robust, median-based center $c_i^{med}$, obtained as the coordinate-wise (or spatial) median of the points in $C_i$. We also replace the within-cluster dispersion $\sigma_i$ by the median of the within-cluster distances to this center, $\sigma_i^{med} = \operatorname{median}_{x \in C_i}\, d(x, c_i^{med})$. The median is robust to outliers, and hence this index is less sensitive to outliers than the classical DBI. The Median Davies-Bouldin Index is then defined as
\begin{equation}
    DBI_M = \frac{1}{K} \sum_{i=1}^K \max_{j \neq i} \left( \frac{\sigma_i^{med} + \sigma_j^{med}}{d\left(c_i^{med}, c_j^{med}\right)} \right). \label{eq:DBIM}
\end{equation}
As with the classical DBI, lower values of $DBI_M$ indicate better-separated, more internally cohesive clusters.

\textbf{Trimmed Calinski-Harabasz Index ($CHI^{trim}$):} We construct a robust version of CHI using trimming, retaining ordinary (mean-based) cluster centers so that the underlying variance decomposition remains exact on the trimmed sample. For each cluster $C_i$, let $c_i$ denote its ordinary sample mean, and define the trimmed subset $C_i^{trim} = \left\{ x \in C_i : d\left(x, c_i\right) \le q_{1-\alpha}\left( d\left(\cdot, c_i\right) \right) \right\}$, where $q_{1-\alpha}(\cdot)$ denotes the $(1-\alpha)$ quantile of within-cluster distances to $c_i$. Writing $c_i^{trim}$ for the sample mean of the retained points in $C_i^{trim}$, and $c^{trim}$ for the grand mean of all retained points pooled across clusters, we define the trimmed between- and within-cluster dispersions as $\operatorname{tr}\left(B_K^{trim}\right) = \sum_{i=1}^K n_i^{trim} \left\lVert c_i^{trim} - c^{trim} \right\rVert^2$ and $\operatorname{tr}\left(W_K^{trim}\right) = \sum_{i=1}^K \sum_{x \in C_i^{trim}} \left\lVert x - c_i^{trim} \right\rVert^2$, where $n_i^{trim} = \left|C_i^{trim}\right|$. Because these quantities are computed from the genuine sample means of the trimmed data, the identity $\operatorname{tr}(B_K^{trim}) + \operatorname{tr}(W_K^{trim}) = \operatorname{tr}(T_K^{trim})$, where $T_K^{trim}$ denotes the total dispersion of the trimmed sample about $c^{trim}$, continues to hold exactly, preserving the interpretability of the resulting ratio as a between/total variance fraction. Writing $n^{trim} = \sum_{i=1}^K n_i^{trim}$ for the total number of retained (non-trimmed) observations, the Trimmed Calinski-Harabasz Index is defined as
\begin{equation}
    CHI^{trim} = \frac{\operatorname{tr}\left(B_K^{trim}\right)}{\operatorname{tr}\left(W_K^{trim}\right)} \times \frac{n^{trim}-K}{K-1}. \label{eq:CHItrim}
\end{equation}
As with the classical CHI, higher values of $CHI^{trim}$ indicate better-defined clusters, while trimming reduces the tendency of a small number of extreme observations to dominate the squared-distance terms that make up both the numerator and the denominator.

\subsubsection{External Evaluation Metrics}

External evaluation metrics gauge the quality of clustering results by comparing them with external information or ground truth labels, in contrast to the internal metrics defined above, which rely only on the data and the resulting partition. Because external metrics require the true class of each observation to be known, they cannot generally be computed in real-world clustering applications where no ground truth exists; however, they provide a more direct measure of clustering accuracy whenever such labels are available, as is the case in our simulated dataset, where the true cluster memberships are known by construction. Included among these metrics are Accuracy and Cohen's $\kappa$. Despite being commonly applied in classification problems, we adopted these measures in our analysis due to the availability of true class labels in the simulated dataset.
  \begin{table}[H]
      \centering
       \begin{tabular}{cc|cc}
    
         \multicolumn{2}{c|}{} & \multicolumn{2}{c} {\textbf{Cluster labels}} \\
         \multicolumn{2}{c|}{} & \textbf{Positive (P)} & \textbf{Negative (N)} \\
         \hline
        \multirow{2}{*}{\textbf{True class labels}} & \textbf{Positive (P)} & $n_1$ & $n_2$ \\
        & \textbf{Negative (N)} & $n_3$ & $n_4$ \\
       
        \end{tabular}
       \caption{The 2x2 truth table for binary classification.}
   \end{table}

\textbf{Accuracy:} Accuracy quantifies the degree of agreement between the cluster labels and the class labels. It is calculated as follows:
\begin{equation}
Accuracy = \dfrac{n_1+n_4}{n_1+n_2+n_3+n_4}. \label{eq:accuracy}
\end{equation}
Superior clustering performance is indicated by higher accuracy values \citep{ahmed2020k}.

\textbf{Cohen's Kappa ($\kappa$):} The Cohen's $\kappa$ coefficient evaluates the concordance between clustering results and the ground truth. It spans from -1 to 1, where +1 signifies perfect agreement, 0 represents agreement by chance, and negative values imply agreement worse than chance:
\begin{equation}
      \kappa = \dfrac{P(A)-P(E)}{1-P(E)}, \label{eq:kappa}
\end{equation}
where $P(A) = \dfrac{n_1+n_4}{n_1+n_2+n_3+n_4}$ and $P(E) = \dfrac{(n_1+n_2) \times (n_1+n_3)+(n_3+n_4)\times(n_2+n_4)}{(n_1+n_2+n_3+n_4)}$.

\subsection{Simulation data example}
The simulations were conducted using the R statistical software, employing the MASS package for generating mixture models with multivariate normal component distributions and drawing samples from such models. The dataset comprises 750 samples, 2 features, and known classes or labels, with 150, 300, and 300 samples drawn from Groups 1, 2, and 3, respectively. In this simulation, three cluster groups were defined, each with their respective means and covariance matrices. Group 1 has a mean vector of (6, 8) and a covariance matrix with diagonal elements of 2 and off-diagonal elements of 1.4. Group 2 has a mean vector of (11, 6) and a covariance matrix with diagonal elements of 2 and off-diagonal elements of $-1.4$. Group 3 has a mean vector of $(0, 11)$ and a covariance matrix 2 times the identity matrix. 30 outliers were introduced into the dataset by appending them to the existing dataset, resulting in a dataset containing the original clusters and the introduced outliers. The developed robust clustering algorithms, MK-means DPD and its convergent variant DC-MK-means DPD, were compared to standard K-means, Mahalanobis K-means (MK-means), DBSCAN, and the Gaussian Mixture Model (GMM). Various parameter settings for generating clusters were considered to assess the performance of each algorithm. The algorithms assigned observations to clusters without considering their known labels. The main focus of the simulation study was to compare the performance of each algorithm in the presence of outliers. The results are presented in Figure \ref{fig_outliers}.

\begin{figure}[H]
\includegraphics[width=1.0\textwidth]{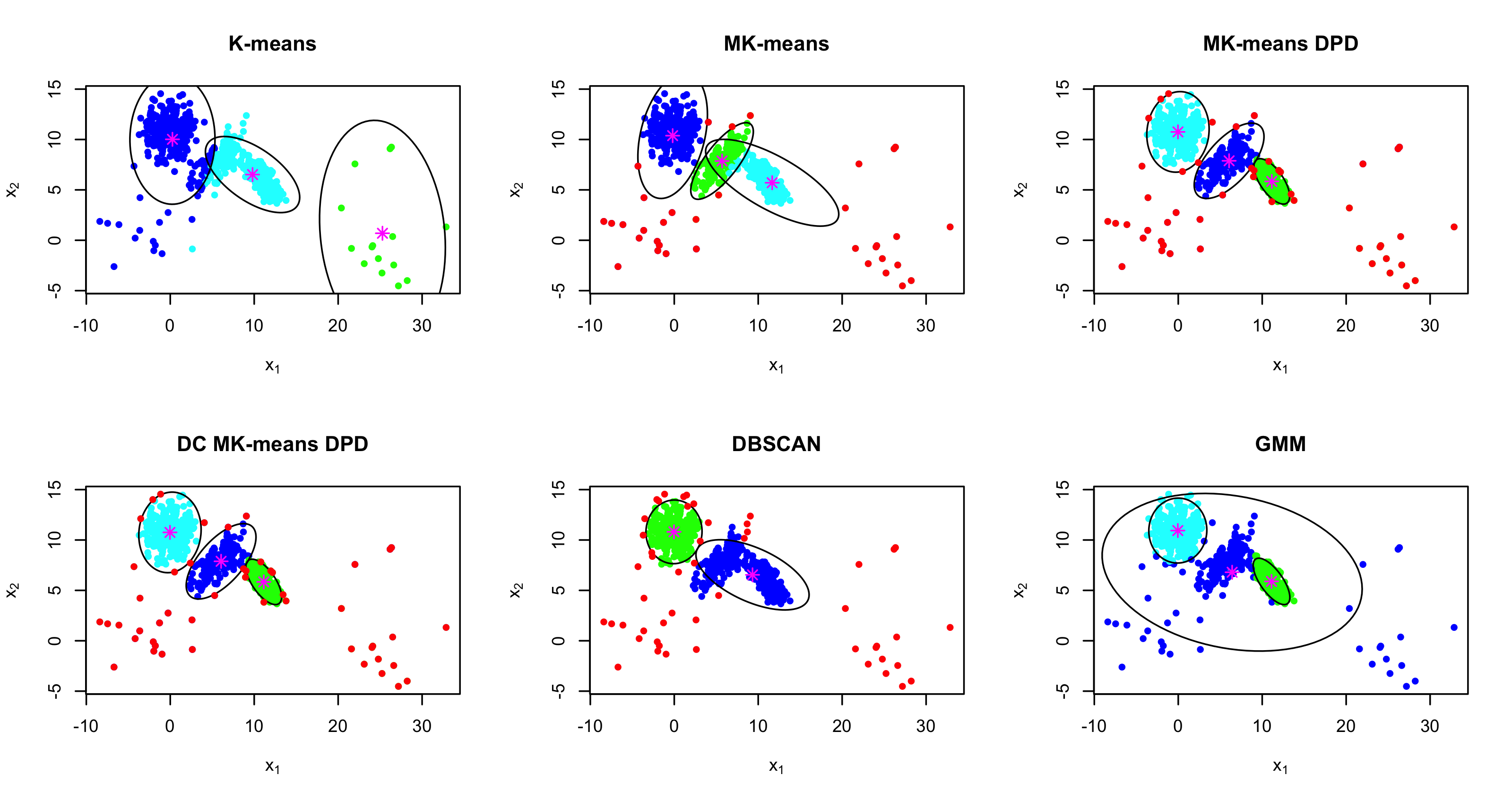} 
  \caption{Plots of clusters from the proposed methods and other existing methods on simulated dataset. Red dots represent outliers detected by the corresponding method, and red stars represent the estimated cluster centers.}
  \label{fig_outliers}
\end{figure}

In Table \ref{table:1} and the corresponding tables in Sections \ref{sec:iris} and \ref{sec:covid}, the \emph{Outlier} column reports the number of observations falling outside the 99\% confidence ellipsoid of their assigned cluster, $(x_i-\hat\mu_k)^T\hat\Sigma_k^{-1}(x_i-\hat\mu_k) > \chi^2_{p,0.01}$, computed using each method's final cluster center and covariance estimates. This provides a concrete illustration of the outlier-detection capability described in the Remark on Advantages of MK-means DPD (Section \ref{sec:mkdpd}).

\begin{table}[h]
\caption{Performance comparison of the proposed methods with existing methods on the simulated dataset.}
\setlength\tabcolsep{3pt}
\begin{center}
\begin{tabular}{ |c|c|c|c|c|c|c|c|c|} 
  \hline
 \diagbox{\textbf{Method}}{\textbf{PM}} & \textbf{Accuracy} & $\boldsymbol{\kappa}$ & $\mathbf{TR^2}$ & $\mathbf{CHI^{trim}}$ &  \textbf{PAMSIL} & $\mathbf{DBI_M}$ & \textbf{Cluster} & \textbf{Outlier}  \\
 \hline
   \textbf{K-means} & 0.788 & 0.660 & 0.801 & 1485.0 &  \textbf{0.735} & \textbf{0.312} & 3 & 0 \\ 
 \hline
   \textbf{MK-means} & 0.950 & 0.923 & 0.900 & 3317.2 & 0.699 & 0.523 & 3 & 35 \\  
 \hline
 \textbf{MK-means DPD} & 0.951 & 0.926 & 0.906 & 3573.4 & 0.704 & 0.537 & 3 & 49 \\  
 \hline
 \textbf{DC MK-means DPD} & \textbf{0.955} & \textbf{0.932} & \textbf{0.908} & \textbf{3645.1} & 0.704 & 0.547 & 3 & 54 \\
 \hline
 \textbf{DBSCAN} & 0.953 & 0.928 & 0.907 & 3614.0 &  0.511 & 0.552 & 3 & 52 \\ 
 \hline
 \textbf{GMM} & 0.938 & 0.907 & 0.791 & 1394.5 & 0.685 & 0.575 & 3 & 0 \\ 
\hline
\end{tabular}
\label{table:1}
\end{center}
\end{table}

To evaluate the proposed algorithms against existing methods, the derived cluster groups were juxtaposed with the known cluster classifications. The percentage of correctly classified observations was then computed for each algorithm. Figure \ref{fig_outliers} clearly illustrates the superior graphical representation of MK-means DPD and DC-MK-means DPD relative to the other methods. As the true cluster groups are known beforehand, we can compare both the external and the robust internal evaluation metrics of the algorithms using this simulated data. Table \ref{table:1} reaffirms that our developed algorithms outperform the other clustering methods considered. DC-MK-means DPD achieves the highest accuracy (0.955), the highest $\kappa$ (0.932), the highest trimmed $R^2$ (0.908), and the highest $CHI^{trim}$ (3645.1) among all methods compared, with MK-means DPD and DBSCAN close behind on every one of these metrics (MK-means DPD: accuracy 0.951, $\kappa$ 0.926, trimmed $R^2$ 0.906, $CHI^{trim}$ 3573.4; DBSCAN: accuracy 0.953, $\kappa$ 0.928, trimmed $R^2$ 0.907, $CHI^{trim}$ 3614.0). Indeed, MK-means DPD and DC-MK-means DPD behave very similarly in terms of overall performance across all of the reported metrics, which is to be expected given that DC-MK-means DPD differs from MK-means DPD only in its assignment criterion, introduced specifically to guarantee convergence rather than to alter the underlying clustering objective. The practical benefit of this modification is instead seen in computational behavior: on this simulated dataset, MK-means DPD required 11 iterations to reach a stable solution, whereas DC-MK-means DPD converged in only 4 iterations, consistent with the finite-convergence guarantee established in Theorem \ref{thm:convergence}. K-means attains the highest PAMSIL (0.735) and the lowest, and therefore best, $DBI_M$ (0.312) among all methods, reflecting the tendency of these internal indices to favor compact, convex cluster shapes even when, as shown by the comparatively lower accuracy and $\kappa$ values for K-means, the resulting partition agrees less well with the true underlying classes. DBSCAN performs competitively with MK-means DPD and DC-MK-means DPD on accuracy, $\kappa$, trimmed $R^2$, and $CHI^{trim}$, correctly recovering all 3 true clusters, but attains a noticeably lower PAMSIL (0.511) than the two proposed methods (0.704 each) and flags a substantial number of outliers (52). As shown by the red dots in Figure \ref{fig_outliers}, MK-means DPD and DC-MK-means DPD successfully detect most of the 30 outliers introduced into the dataset, illustrating the outlier-detection capability described in Section \ref{sec:mkdpd}. Taken together, these results support the robustness of the proposed DPD-based approach to outlier contamination, while showing that the additional structural modification introduced in DC-MK-means DPD achieves the same strong clustering performance as MK-means DPD, with substantially faster convergence, and remains competitive with, or superior to, DBSCAN across all metrics considered.

\section{Real data application}

Having established the performance of the proposed MK-means DPD algorithm on simulated data with known cluster memberships and artificially introduced outliers, we now turn to two real-world applications in which the true underlying structure is not fully known in advance and outliers arise naturally from the data itself. We first apply the method to the well-known Iris dataset, for which the true species labels are available and can be used to assess clustering accuracy directly. We then apply the method to Covid-19 case fatality rate, infection rate, and related demographic and socio-economic data for countries around the world, a genuinely real-world setting containing natural outliers, in order to examine the resulting clusters and their geographic distribution.

\subsection{Application to Iris data}
\label{sec:iris}

The Iris dataset, sourced from Scikit-learn, a machine learning library, consists of 150 samples distributed across three species: Setosa, Versicolor, and Virginica. Each class comprises 50 sample data points. The dataset features four numeric attributes denoting measurements in centimeters, namely Sepal length, Sepal width, Petal length, and Petal width. The fifth attribute is qualitative, signifying the class name corresponding to the plant species. This allows us to treat the data as if the true clusters are known, so that external evaluation metrics can be computed. The objective is to evaluate and compare the performance of the two proposed robust methods alongside existing clustering methods, assessing their robustness and effectiveness in the presence of outliers.

The assessment involves visual inspection and comparison of several performance metrics as defined in the previous section. The two parameters required for our modified MK-means DPD algorithm are the number of clusters, K, and the DPD parameter, $\alpha$. K is chosen to be 3 based on the 3 known species in the dataset. For this analysis, $\alpha$ is set to 0.2. Both external and internal evaluation metrics were used to compare the two new methods with the existing ones. The resulting cluster plots are shown in Figures \ref{fig: width} and \ref{fig: length}.

\begin{center}
\begin{figure}[h]
\includegraphics[width=1\textwidth]{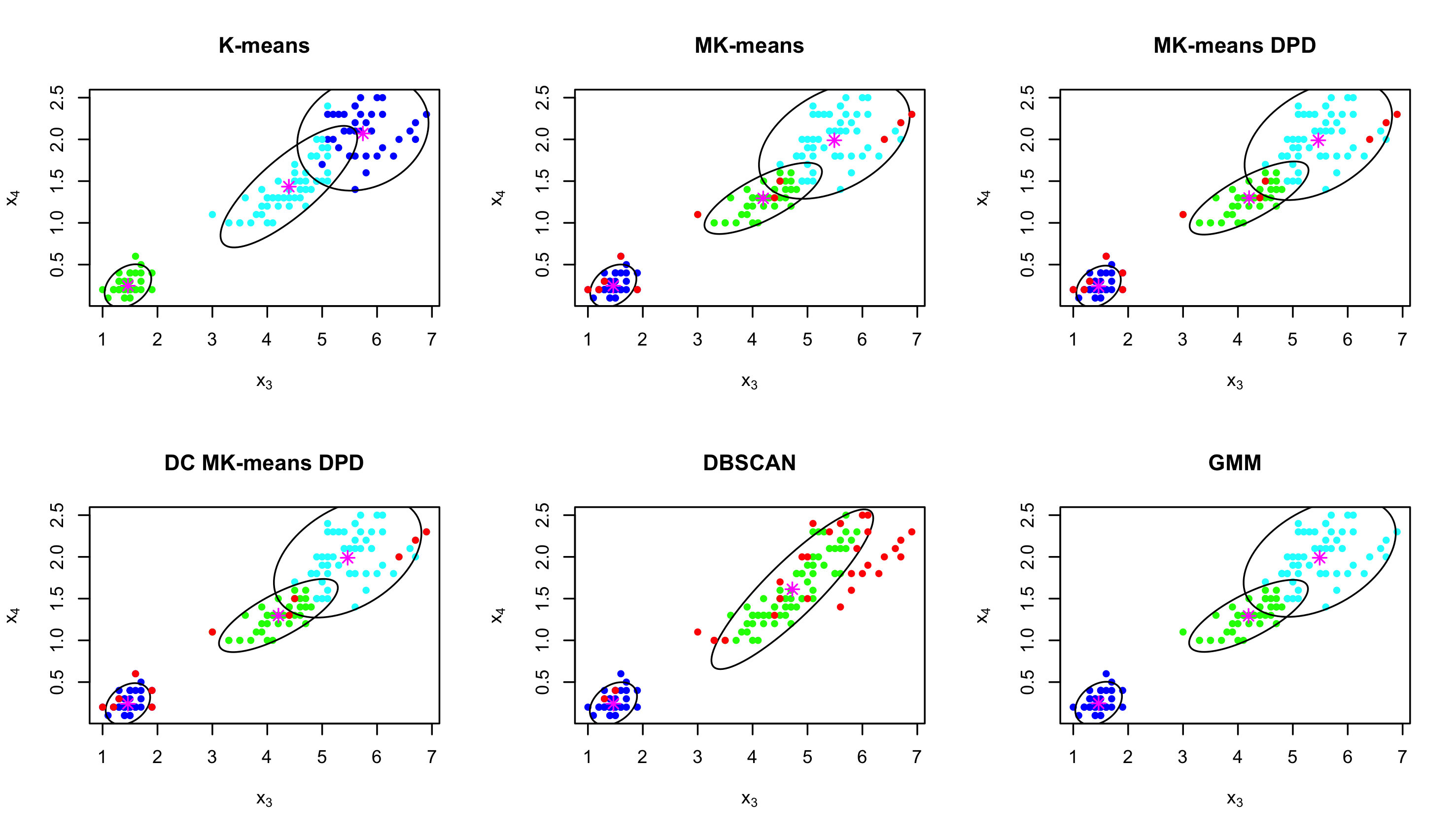}
  \caption{Plots of clusters showing petal length and petal width, as obtained from various algorithms using Iris data. Red dots represent outliers detected by the corresponding method, and red stars represent the estimated cluster centers.}
    \label{fig: width}
  \end{figure}
\end{center}

\begin{center}
\begin{figure}[h]
\includegraphics[width=1.0\textwidth]{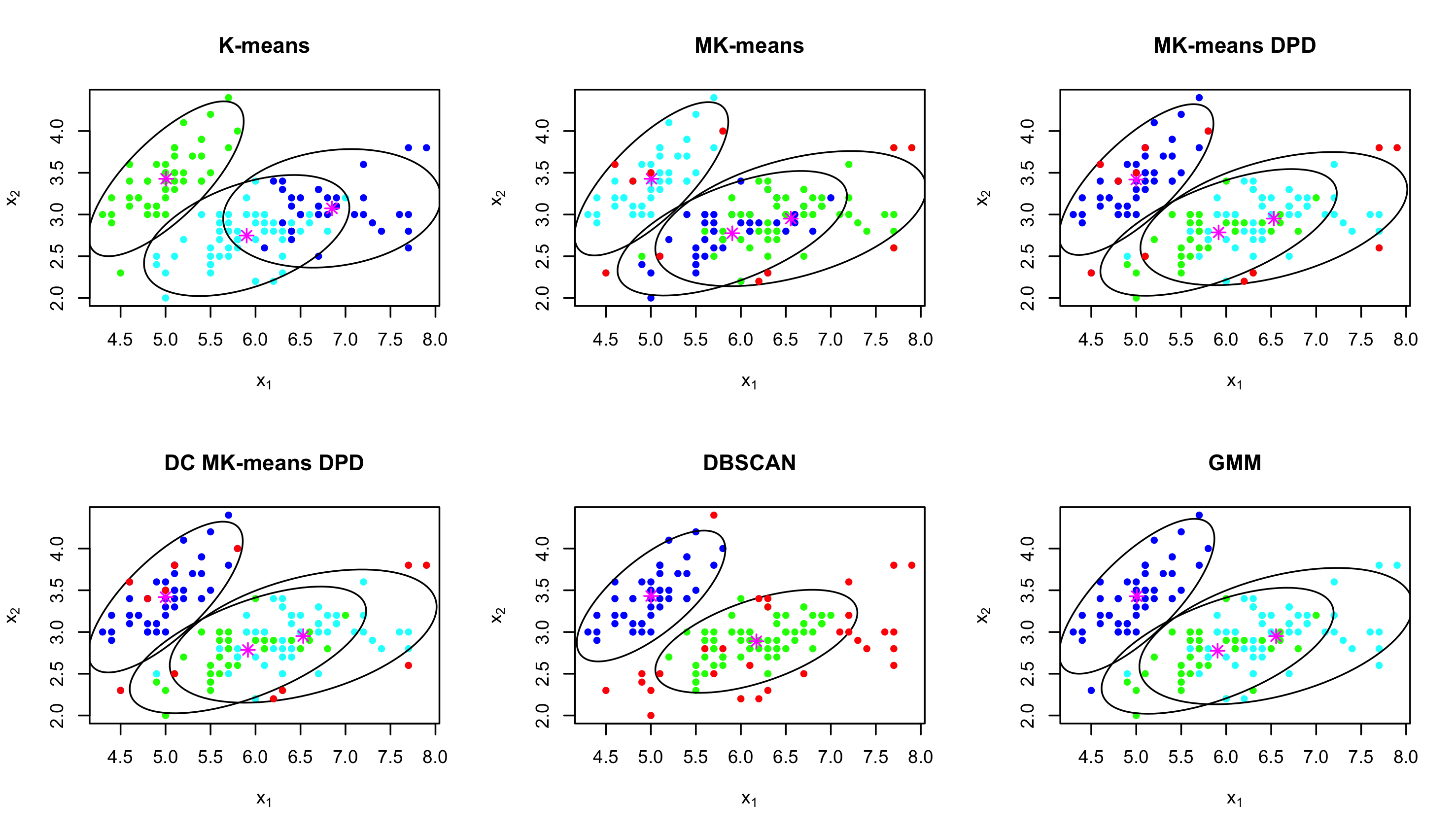}
  \caption{Plots of clusters showing sepal length and sepal width, as obtained from various algorithms using Iris data. Red dots represent outliers detected by the corresponding method, and red stars represent the estimated cluster centers.}
    \label{fig: length}
  \end{figure}
\end{center}

\begin{table}[H]
\caption{Performance comparison of the proposed methods with existing methods on Iris data.}
\setlength\tabcolsep{3pt}
\begin{center}
\begin{tabular}{ |c|c|c|c|c|c|c|c|c|} 
  \hline
 \diagbox{\textbf{Method}}{\textbf{PM}} & \textbf{Accuracy} & $\boldsymbol{\kappa}$ & $\mathbf{TR^2}$ & $\mathbf{CHI^{trim}}$ &  \textbf{PAMSIL} & $\mathbf{DBI_M}$ & \textbf{Cluster} & \textbf{Outlier}  \\
 \hline
   \textbf{K-means} & 0.833 & 0.750 & \textbf{0.821} & \textbf{318.5} &  0.574 & 0.834 & 3 & 0 \\ 
 \hline
   \textbf{MK-means} & 0.960 & 0.940 & 0.738 & 195.9 & 0.487 & 0.957 & 3 & 11 \\  
 \hline
 \textbf{MK-means DPD} & \textbf{0.967} & \textbf{0.950} & 0.740 & 197.8 & 0.494 & 0.945 & 3 & 12 \\  
 \hline
 \textbf{DC MK-means DPD} & \textbf{0.967} & \textbf{0.950} & 0.740 & 197.8 & 0.494 & 0.945 & 3 & 12 \\
 \hline
 \textbf{DBSCAN} & 0.680 & 0.520 & 0.687 & 152.3 &  \textbf{0.624} & \textbf{0.540} & 2 & 4 \\ 
 \hline
 \textbf{GMM} & \textbf{0.967} & \textbf{0.950} & 0.793 & 265.7 & 0.494 & 0.945 & 3 & 0 \\ 
\hline
\end{tabular}
\label{table:2}
\end{center}
\end{table}

In Figures \ref{fig: width} and \ref{fig: length}, MK-means DPD, DC-MK-means DPD, and the Gaussian Mixture Model appear visually superior, consistent with their accuracy and $\kappa$ values reported in Table \ref{table:2}. MK-means DPD, DC-MK-means DPD, and GMM all achieve the same highest accuracy (0.967) and $\kappa$ (0.950) among the methods compared, with MK-means close behind (accuracy 0.960, $\kappa$ 0.940). K-means performs noticeably worse in terms of accuracy (0.833) and $\kappa$ (0.750), while nonetheless attaining the highest trimmed $R^2$ (0.821) and the highest $CHI^{trim}$ (318.5) among all methods considered.

This pattern, in which K-means and the DPD-based methods trade off which metrics they lead, differs from what was observed in the simulation study of Section 4.2, where DC-MK-means DPD dominated essentially every metric. A likely explanation is that the Iris dataset, unlike the simulated dataset, contains very few genuine outliers or points of substantial ambiguity between species. In the near-absence of contamination, the robustness of MK-means DPD and DC-MK-means DPD confers only a modest additional benefit over classical methods, with MK-means DPD, DC-MK-means DPD, and GMM achieving similar, though not identical, performance on the metrics reported in Table \ref{table:2}.

Although DBSCAN attains the highest PAMSIL (0.624) and the lowest, and therefore best, $DBI_M$ (0.540) among all methods, its accuracy (0.680) and $\kappa$ (0.520) are considerably lower than those of the other methods, since it identifies only 2 of the 3 true species, failing to distinguish the Versicolor and Virginica species, whose sepal and petal measurements are known to overlap substantially. This illustrates that internal and external metrics do not always agree, underscoring the value of considering both together.

\subsection{Application to Covid-19 data}
\label{sec:covid}

Next, we apply the proposed robust MK-means DPD method to cluster countries based jointly on Covid-19 case fatality rate (CFR), infection rate, and a set of demographic and socio-economic variables, namely GDP per capita, number of hospital beds, infant mortality rate, life expectancy, and the percentage of the population aged 65 and older, drawn from the countries of the world dataset. The primary objective of this cluster analysis is to identify coherent groupings of countries that share similar pandemic outcomes together with similar demographic and healthcare characteristics, with GDP per capita of particular interest as a convenient measure of a nation's wealth. More broadly, this clustering study is of practical importance, as it can reveal empirical features and relationships among case fatality rate, infection rate, and the demographic and healthcare characteristics of countries that may not otherwise be apparent, thereby offering insight into how factors such as wealth, healthcare capacity, and population age structure relate to the way the pandemic affected different countries.

This analysis involved two primary datasets. The initial dataset, sourced from the Johns Hopkins COVID database (https://www.kaggle.com/datasets/antgoldbloom/covid19-data-from-john-hopkins-university), provided daily data on reported cases and deaths in each country. Additional variables, including CFR and infection rate, were derived from this dataset, with CFR calculated as the ratio of deaths to cases multiplied by 100 and infection rate obtained by dividing total cases by the country's population. The second dataset, comprising health-related and demographic information for countries around the world, included variables such as GDP per capita, literacy percentage, phone count, land usage, climate region, economic industries, and overall death rate, and is available from the following sources:\\
https://www.kaggle.com/code/mehmettek/data-science-with-world-countries/data?select=countries+of+the+world.csv.\\
https://www.kaggle.com/code/mehmettek/data-science-with-world-countries/notebook

However, as shown in Figures \ref{cfr} and \ref{inf}, both CFR and infection rate contain a substantial number of outliers, driven in large part by small, densely populated nations such as Singapore and Monaco, where the entire country effectively functions as a single dense urban area. In the presence of such outliers, traditional clustering methods, including the standard K-means and Mahalanobis-distance K-means algorithms, are prone to distorted estimates of cluster centers and covariance matrices, and the resulting clusters can consequently portray a substantially different, and potentially misleading, picture of the underlying relationships in the data. This is precisely the setting for which the proposed robust DC-MK-means DPD method is designed, and applying it here provides a meaningful real-world demonstration of its robustness. The clusters obtained from the robust method are displayed in the spider plot in Figure \ref{fig: covid star}, which we use to better understand the relationships and trends present in this data, and in the world map in Figure \ref{fig: covid world}, which shows the geographic distribution of the resulting clusters across countries.

\begin{figure}[h]
    \centering
    \begin{minipage}{0.48\textwidth}
        \centering
        \includegraphics[width=0.65\textwidth]{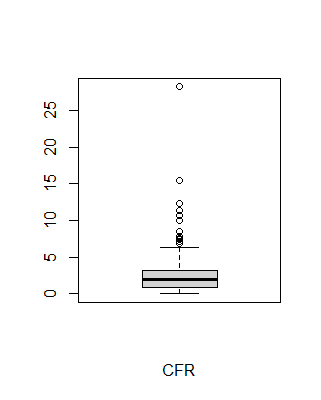}
        \caption{Boxplot of Covid-19 CFR across different countries}
        \label{cfr}
    \end{minipage}\hfill
    \begin{minipage}{0.48\textwidth}
        \centering
        \includegraphics[width=0.65\textwidth]{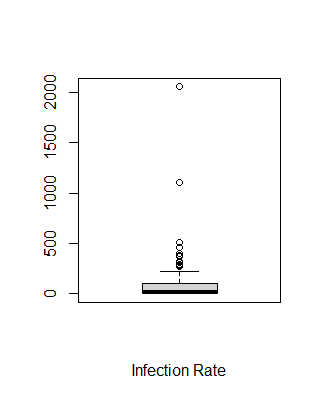}
        \caption{Boxplot of Covid-19 infection rate across different countries}
        \label{inf}
    \end{minipage}
\end{figure}

Since the true clusters of countries are not known in this real-world setting, external evaluation metrics cannot be computed, and we instead compare the candidate clustering methods using the robust internal evaluation metrics introduced in Section \ref{sec:robust_metrics}. Table \ref{table:3} reports the results. DC-MK-means DPD achieves the highest trimmed $R^2$ (0.672), the highest $CHI^{trim}$ (173.9), and the lowest, and therefore best, $DBI_M$ (0.875) among all methods compared, with MK-means DPD close behind on trimmed $R^2$ and $CHI^{trim}$; together, the two proposed methods clearly outperform the remaining methods on these metrics. K-means attains the highest PAMSIL (0.455), with DC-MK-means DPD close behind (0.434); the remaining methods report somewhat lower PAMSIL values, ranging from 0.323 to 0.384. Motivated by these results, the remainder of this section focuses on the clusters obtained from the proposed DC-MK-means DPD method.

\begin{table}[h]
\caption{Performance comparison of the proposed methods with existing methods on the Covid-19 dataset.}
\setlength\tabcolsep{3pt}
\begin{center}
\begin{tabular}{ |c|c|c|c|c|c|c|} 
  \hline
 \diagbox{\textbf{Method}}{\textbf{PM}} & $\mathbf{TR^2}$ & $\mathbf{CHI^{trim}}$ &  \textbf{PAMSIL} & $\mathbf{DBI_M}$ & \textbf{Cluster} & \textbf{Outlier}  \\
 \hline
   \textbf{K-means} & 0.649 & 156.9 & \textbf{0.455} & 0.989 & 3 & 0 \\ 
 \hline
   \textbf{MK-means} & 0.562 & 108.9 & 0.384 & 1.079 & 3 & 17 \\  
 \hline
 \textbf{MK-means DPD} & 0.660 & 165.2 & 0.383 & 1.127 & 3 & 47 \\  
 \hline
 \textbf{DC MK-means DPD} & \textbf{0.672} & \textbf{173.9} & 0.434 & \textbf{0.875} & 3 & 51 \\
 \hline
 \textbf{DBSCAN} & 0.266 & 30.8 & 0.323 & 2.324 & 2 & 68 \\ 
 \hline
 \textbf{GMM} & 0.474 & 76.7 & 0.354 & 1.435 & 3 & 0 \\ 
\hline
\end{tabular}
\label{table:3}
\end{center}
\end{table}

\begin{figure}[h]
\centering
\includegraphics[width=0.9\textwidth]{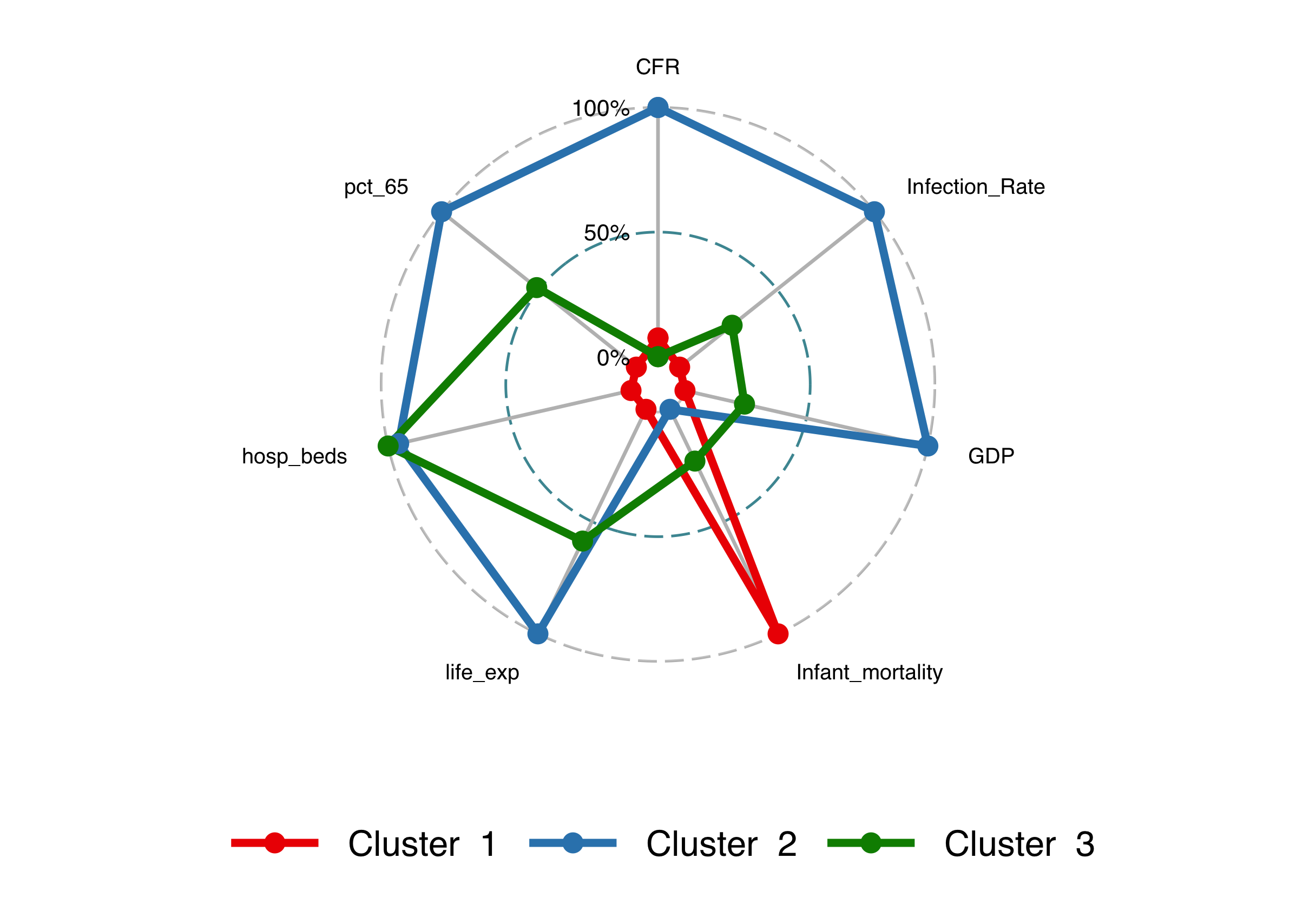}
  \caption{Spider (radar) plot comparing the three country clusters identified by DC-MK-means DPD across seven variables: CFR, infection rate, GDP per capita, infant mortality, life expectancy, number of hospital beds, and the percentage of the population aged 65 and above.}
  \label{fig: covid star}
\end{figure}

Next, we study the relationship between the demographic variables, CFR, and infection rate across the identified clusters using a spider plot (Figure \ref{fig: covid star}). In a spider plot, each spoke represents one variable, rescaled so that 0\% and 100\% denote the lowest and highest values observed across all countries, respectively, with a cluster's polygon extending further outward on a given axis indicating a higher value on that variable relative to other countries.

Viewed this way, the spider plot in Figure \ref{fig: covid star} shows that DC-MK-means DPD identifies three well-separated, internally coherent country profiles rather than arbitrary or overlapping groupings. Cluster 2 combines high CFR, infection rate, GDP per capita, life expectancy, and elderly population share with low infant mortality, consistent with wealthier, more developed nations; its elevated CFR despite strong healthcare infrastructure is plausibly explained by its older population, since COVID-19 mortality is well known to rise sharply with age. Cluster 3 combines the lowest CFR with the highest hospital-bed rate among the three clusters, an intuitive pairing suggesting that greater healthcare capacity is associated with lower fatality, independent of the wealth and age effects driving Cluster 2. Cluster 1 is low across nearly every variable except infant mortality, where it is highest, consistent with lower-income countries combining limited healthcare capacity with higher baseline infant mortality; its low CFR and infection rate may partly reflect limited testing and reporting capacity rather than a genuinely lower disease burden. Overall, the distinct, non-overlapping shape of each cluster indicates that these seven variables jointly define coherent, well-separated groups of countries, rather than producing three near-identical profiles.

\begin{figure}[h]
\centering
\includegraphics[width=0.9\textwidth]{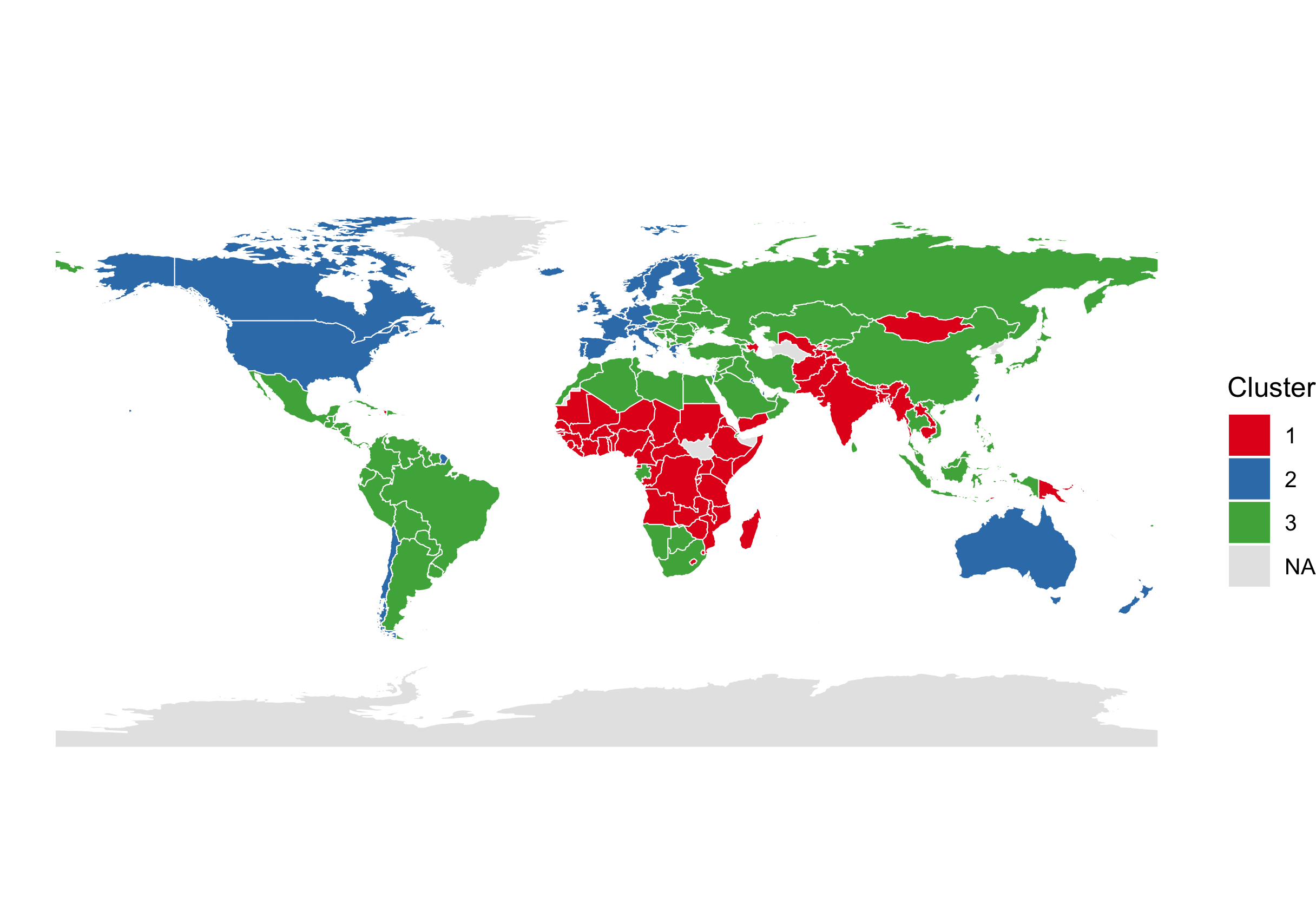}
  \caption{World map colored by the three clusters identified by the DC-MK-means DPD algorithm, with countries not included in the analysis (due to missing data) shown in gray.}
  \label{fig: covid world}
\end{figure}

The geographic distribution of the clusters, shown in Figure \ref{fig: covid world}, offers a complementary view to the spider plot in Figure \ref{fig: covid star} and reinforces the interpretation given above. Cluster 2 (the high-CFR, high-infection, high-GDP, aging-population profile) is concentrated in North America, Australia, New Zealand, and parts of Northern and Western Europe, consistent with its association with wealthier and more developed nations. Cluster 3 (the low-CFR, high-hospital-capacity profile) covers most of Europe, Russia, China, much of South America, and parts of the Middle East and Southern Africa, most of which have relatively well-developed healthcare systems. Cluster 1 (countries with comparatively less socio-economic development and elevated infant mortality) is concentrated across most of Africa, South Asia (including India, Pakistan, and Bangladesh), and parts of Southeast and Central Asia. This broad geographic coherence, with each cluster corresponding to a recognizable region or set of regions rather than being scattered randomly across the globe, provides further support for the validity of the clustering solution, since geographic location was not itself used to form the clusters.

\section{Conclusion}

We introduced two robust K-means clustering algorithms, MK-means DPD and its convergent variant DC-MK-means DPD, incorporating Mahalanobis distance as the distance metric and density power divergence (DPD) measures to estimate cluster centers and covariance matrices at each iteration. This approach makes the algorithms less sensitive to outliers or extreme values compared to traditional clustering methods, while the inclusion of Mahalanobis distance ensures adaptability to heterogeneous and non-spherical clusters. Mahalanobis distance-based K-means clustering, on which our methods are built, is known to lack a general convergence guarantee; to address this, we introduced DC-MK-means DPD, which redefines the cluster assignment step in terms of a pointwise density power divergence loss, and established a formal theorem (Theorem \ref{thm:convergence}) guaranteeing that this variant converges in a finite number of iterations. This theoretical guarantee was reflected empirically in our simulation study, where MK-means DPD required 11 iterations to converge, compared to only 4 for DC-MK-means DPD, while the two algorithms achieved comparable clustering performance across all reported metrics. We also introduced two new robust internal evaluation indexes, the Median Davies-Bouldin Index and the Trimmed Calinski-Harabasz Index, to complement existing robust metrics such as trimmed $R^2$ and PAMSIL, ensuring that the performance comparisons reported throughout this article are not themselves distorted by outliers. Through a simulation study, we demonstrated the strong performance of our proposed algorithms relative to four commonly used clustering methods: K-means, Mahalanobis K-means, DBSCAN, and Gaussian Mixture Models. Subsequently, we applied the methodology to two real datasets: 1) Iris data and 2) Covid-19 data. The Iris application, in which genuine outliers are scarce, showed that the robustness of MK-means DPD and DC-MK-means DPD confers only a modest advantage over classical methods in the near-absence of contamination, reinforcing the broader picture established by the simulation study and the Covid-19 application: the benefit of the proposed DPD-based approach is most pronounced specifically when outliers are present in the data.

While our simulation study and real-world applications assumed a known number of clusters, it is worth noting that the developed initialization scheme (Section \ref{sec:init}) also has the potential to help identify the optimal number of clusters within a given dataset, and we plan to explore this aspect of the methodology in greater detail in future studies. An additional direct application of the proposed algorithms is their ability to aid in the identification of outliers, addressing a challenge often encountered in high-dimensional data. As with classical K-means, the convergence guarantee established for DC-MK-means DPD is to a stationary point rather than a global optimum, and extending these theoretical results, as well as developing a principled, data-driven procedure for selecting the DPD tuning parameter $\alpha$, remain natural directions for future work.

The application to the Covid-19 dataset illustrated the practical value of the proposed robust clustering approach in a real-world setting containing substantial, naturally occurring outliers. The clusters obtained using DC-MK-means DPD, based jointly on case fatality rate, infection rate, and a set of demographic and healthcare variables (GDP per capita, hospital beds, infant mortality, life expectancy, and population age structure), corresponded to coherent and interpretable groupings of countries that also exhibited a clear geographic pattern, despite geographic location not itself being used to form the clusters. This finding serves as motivation for future research in this area; our specific plans include conducting more in-depth investigations into how these clusters relate to other country-level factors not used in the clustering itself, such as literacy rates, climate region, and economic structure, using regression and supervised learning models within each of the identified clusters.

\bibliography{Clustering_Reference}

\end{document}